\documentclass[a4paper,11pt]{article}
\pdfoutput=1

\usepackage{jheppub}
\usepackage{natbib}
\usepackage{braket}
\usepackage{xcolor}
\usepackage{textcomp}
\usepackage{dcolumn}
\usepackage{bm}
\usepackage{cleveref}
\usepackage{algorithm}
\usepackage{algpseudocode}
\usepackage[utf8]{inputenc}
\usepackage{tcolorbox}
\usepackage{todonotes}
\usepackage{amssymb,amsmath,amsbsy,amsthm,mathtools}
\usepackage{braket}
\usepackage{graphicx,color}
\usepackage{multirow}
\usepackage{hyperref}
\usepackage{subcaption}
\usepackage{adjustbox}
\usepackage{relsize} 
\usepackage{xcolor}
\usepackage{booktabs}
\usepackage{subcaption}
\usepackage{array}

\usepackage{slashed}

\newcolumntype{P}[1]{>{\centering\arraybackslash}p{#1}}
\newcolumntype{M}[1]{>{\centering\arraybackslash}m{#1}}

\newtheorem{definition}{Definition}[section]
\newtheorem{proposition}{Proposition}[section]
\newtheorem{lemma}{Lemma}[section]

\newtheorem{conjecture}{Conjecture}[section]

\newcommand{\mbE}{\mathbb{E}}

\newcommand{\mbZ}{\mathbb{Z}}

\usepackage{graphicx}

\title{Tripartite Correlation Signal from Multipartite Entanglement of Purification}

\author[a,b]{Ning Bao}
\author[a]{Keiichiro Furuya}
\author[a,c]{Joydeep Naskar}

\affiliation[a]{Department of Physics, Northeastern University, Boston, MA, 02115, USA}
\affiliation[b]{Computational Science Initiative, Brookhaven National Laboratory, Upton, NY 11973 USA}
\affiliation[c]{The NSF AI Institute for Artificial Intelligence and Fundamental Interactions, Cambridge, MA, U.S.A.}

\emailAdd{ningbao75@gmail.com}
\emailAdd{k.furuya@northeastern.edu}
\emailAdd{naskar.j@northeastern.edu}

\abstract{We propose a signal $\Delta^{(3)}_p$ for tripartite correlations in finite-dimensional quantum systems and $\Delta^{(3)}_w$ for holographic systems. We prove that $\Delta^{(3)}_p$ is non-negative for any tripartite entangled mixed states and vanishes for mixed product states. It is a correlation signal because it is generally nonzero for a separable state containing a classical mixture. Based on the conjecture, the equality between an entanglement wedge cross section $E_w$ and entanglement of purification $E_p$, i.e., $E_w = E_P$ in the semiclassical limit, we apply the tripartite correlation signal to study the structures of tripartite entanglement in AdS$_3$/CFT$_2$, especially for pure AdS$_3$. We comment on a generalization to $n$-partite correlation signals $\Delta^{(n)}_p(A_1:\cdots:A_n)$.}

\makeatletter
\gdef\@fpheader{}
\makeatother
\begin{document}

\maketitle

\section{Introduction}
Quantum entanglement is an ubiquitous phenomena of nature and have been extensively studied in the past century (e.g., see review \cite{quantum-entanglement-review}). In recent decades, the science and engineering of quantum computing have increased the focus on quantum entanglement, especially since entanglement is used as a resource in quantum computing. Both bipartite and multipartite entanglement have been under the lens of investigation. While we understand bipartite entanglement entropy very well, multipartite entanglement is relatively less understood. There have been several approaches towards defining a measure for multipartite entanglement (e.g., see review \cite{Ma:2023ecg}), and it is still an active research area.

The \texttt{It From Qubit} program has demonstrated the fundamental importance of quantum entanglement in the emergence of space-time itself \cite{VanRaamsdonk:2010pw,Swingle:2014uza,Swingle:2017blx,Swingle:2012wq,Bao:2019bib,Cao:2020uvb}, and the $AdS/CFT$ correspondence \cite{Maldacena:1997re} is an excellent laboratory for theoretical investigations of gravity through quantum entanglement\cite{Ryu:2006ef,Maldacena:2013xja,Pastawski:2015qua,Brown:2015bva,Bao:2015bfa}. It has been suggested in \cite{Balasubramanian:2014hda}, that the study multipartite entanglement in holography is inherently relevant for understanding multi-boundary wormholes, and at least tripartite entanglement is necessary for entanglement wedge cross-section \cite{Akers:2019gcv,Hayden:2021gno}. Some recent paradigms which attempt to understand multipartite entanglement in holography, include but are not limited to, the holographic entropy cone \cite{Bao:2015bfa,Bao:2024azn,Hernandez-Cuenca:2022pst,Avis:2021xnz,Czech:2021rxe}, multipartite generalizations of entanglement of purification \cite{Bao:2018fso,Bao:2019zqc,Bao-2019-conditionalEoP}, bit-threads\cite{Cui:2018dyq,Harper:2019lff}, multi-entropy \cite{Gadde:2022cqi,Gadde:2023zni,Gadde:2023zzj,Gadde:2024jfi,Gadde:2024taa,Harper:2024ker}, upper bounding entanglement entropy correlations\cite{Ju:2024hba,Ju:2024kuc,Ju:2025tgg}, and many others.

However, there is no widely accepted notion of a measure of multipartite entanglement in gravity (there are several proposals, e.g., see \cite{Gadde:2023zzj,Penington:2022dhr}). We will take a step back and ask ourselves, can we find reliable signals of multipartite entanglement in holography? A signal doesn't have to satisfy various properties that a measure must\cite{Ma:2023ecg}. A good signal is semi-definite in nature, where it can definitively indicate the presence of genuine multipartite entanglement with a non-zero value. Some efforts have been made in this direction as well (see, e.g., \cite{Balasubramanian:2024ysu}). In this work, we will define a signal of tripartite correlation for mixed states using multipartite entanglement of purification (MEoP) and its holographic version using multipartite entanglement wedge cross-section\cite{Bao:2019zqc,Bao-2019-conditionalEoP}. We found that the signal is non-zero for a tripartite entangled state and zero for a factorized mixed state. However, our definition of the signal does not vanish for some quantum states with classical correlations. In this sense, it is a tripartite correlation signal rather than a tripartite entanglement signal. 

\subsection*{Related Works and Organization}
Several related works appeared in the direction of signals and measures of genuine multipartite entanglement, in particular, we will list three of them below,
\begin{enumerate}
    \item In \cite{Balasubramanian:2024ysu}, three different signals of multipartite entanglement $I_n$, $R_n$, $Q_n$ are defined, but none of them are positive definite for all quantum states with true multipartite entanglement, i.e., the zero set of these signals also includes those quantum states that have genuine multipartite entanglement. 
    
    Our work is yet another variant of this work. The main distinction between their work and ours lies in the use of \textit{reflected entropy} and \textit{entanglement of purification} (EoP). Reflected entropy and EoP differ mainly in three points. First, the computation of the reflected entropy is simpler than that of the EoP because the latter requires an optimization over a purification\cite{terhalEoP2002}. Second, the reflected entropy is not necessarily a correlation measure for some quantum states\cite{Hayden2023REnotCorrelation}, i,e., it does not always satisfy the monotonicity under a partial trace. In contrast, the EoP is a correlation measure for any quantum state. Third, a canonically purified state can contain residual entanglement between a system and a purifier, whereas the optimal purified state that computes the EoP for a given density matrix contains only minimal entanglement to realize the purification. 
     
    Given these differences, we are motivated to construct signals for multipartite entanglement from EoPs as they did with reflected entropies\footnote{See \cite{Louisia:2025bxz,Mori:2025gqe} for the studies on multipartite correlation measures, involving multipartite generalizations of quantum discord and entanglement wedge cross section, and their inequalities in holography.}.
    
    In a recent work \cite{Balasubramanian:2025hxg}, it was claimed that purely GHZ-type states are forbidden in holography. This renders the proposed signals to be more useful in holography.

    \item In \cite{Iizuka:2025ioc,Iizuka:2025caq}, the authors have defined a measure using multi-entropy (a multipartite information quantity defined in \cite{Gadde:2022cqi}) by subtracting lower-partite ($k<n$) multi-entropy contributions to the $n$-party multi-party entropy. Our method uses a similar subtraction scheme, however, the relevant quantity that we use is multipartite entanglement wedge cross-section. Some more recent work on genuine multi-entropy can be found in \cite{Berthiere:2025toi,Harper:2025uui}.

    \item In \cite{Basak:2024uwc}, a quantity called latent entropy (L-entropy) is defined as a measure of multipartite entanglement using appropriate combination products and roots of bipartite entanglement quantities. However, our signal is a linear combination, and not a product.
\end{enumerate}

The organization of this paper is as follows. In section \ref{sec:density-matrix}, we introduce a definition of a tripartite correlation signal $\Delta_p^{(3)}$ for a density matrix. We compute the quantity in terms of a $GHZ_N$ state and a $W_N$ state. In section \ref{sec:holography}, we propose a definition of a tripartite correlation signal $\Delta_w^{(3)}$ for a holographic state. We study the quantity for pure AdS$_3$, and make brief comments in the case of a BTZ balck hole in section \ref{sec:discussions}. In section \ref{sec:delta_n}, we comment about the multipartite generalization of our tripartite quantity. In section \ref{sec:discussions}, we discuss potential applications and future work. \newline


\noindent
\textbf{Notation}
\begin{itemize}
    \item For a density matrix $\rho_A=\sum_i p_i \rho_i$, $S_A$ and $H(\{p\}_{\rho_A})$ are the von Neumann entropy of $\rho_A$ and the Shannon entropy of $\{p\}_{\rho_A} : =\{p_i\}$ respectively. 

    \item $E_p(A:B)$ is the entanglement of purification introduced in \cite{terhalEoP2002}. $E^{(n)}_p$ for $n\geq 2$ are the multipartite entanglment of purification defined in \eqref{eq:E_p_n}\cite{Umemoto-2018-MultipartiteEoP}.

    \item $E_w(A:B)$ is the entanglement wedge cross-section. $E^{(n)}_w$ for $n\geq 2$ are the multipartite entanglment wedge cross sections defined in \eqref{eq:E_W_N}\cite{Umemoto-2018-MultipartiteEoP}.


\end{itemize}

\section{Tripartite correlation signal $\Delta_p^{(3)}$ for density matrices}\label{sec:density-matrix}

We will construct and study the properties of tripartite correlation signal $\Delta^{(3)}_p(A:B:C)$ for a tripartite density matrix $\rho_{ABC}$ based on the multipartite entanglement of purification (MEoP) \cite{Umemoto-2018-MultipartiteEoP}. The MEoP is a multipartite generalization of the entanglement of purification (EoP) $E_p(A:B)$ proposed in \cite{terhalEoP2002}. $E_p(A:B)$ is defined as
\begin{equation} \label{eq:EoP}
    E_p(A:B) = \underset{\ket{\psi}_{A\tilde{A}B\tilde{B}}}{min} [S_{A\tilde{A}}]
\end{equation}
for a bipartite density matrix $\rho_{AB}$ and a purifier $\tilde{A}\tilde{B}$. We begin by reviewing the definition and the properties of MEoP.

\begin{definition}[Multipartite entanglement of purification(MEoP)\cite{Umemoto-2018-MultipartiteEoP}]

    For $n$-partite density matrices $\rho_{A_1\cdots A_n}$, the multipartite entanglement of purification(MEoP) is defined as
    \begin{equation}\label{eq:E_p_n}
        E_p^{(n)}(A_1: \cdots : A_n) := \underset{\ket{\psi}_{A_1\tilde{A}_1\cdots A_n\tilde{A}_n}}{min} [S_{A_1\tilde{A}_1}+\cdots +S_{A_n\tilde{A}_n}].
    \end{equation}
    where the minimization is taken over any purifications $\ket{\psi}_{A_1\tilde{A}_1\cdots A_n\tilde{A}_n}$ of $\rho_{A_1\cdots A_n}$.
\end{definition}

For the case of a bipartite density matrix $\rho_{AB}$, the EoP $E_p(A:B)$ is different from the \textit{bipartite EoP} $E^{(2)}_p(A:B)$ by a factor of $2$, i.e.,
\begin{equation}\label{eq:Ep2=2EoP}
    E^{(2)}_p(A:B) = \underset{\ket{\psi}_{A\tilde{A}B\tilde{B}}}{min}[S_{A\tilde{A}}+S_{B\tilde{B}}]  = 2 E_p(A:B).
\end{equation}
Here, we used the purification symmetry, i.e., $S_{A\tilde{A}}=S_{B\tilde{B}}$ at the last equality. The basic properties of the MEoP are listed in appendix \ref{app:MEoP}. 

\subsection{Definition and properties} 

We now propose a genuine tripartite correlation signal denoted by $\Delta_p^{(3)}(A:B:C)$ for a mixed state $\rho_{ABC}$. It indicates the existence of some genuine tripartite entanglement and classical correlations of a mixed quantum state by a positive value of $\Delta^{(3)}_p(A:B:C)$, and zero otherwise. As we will see later, $\Delta^{(3)}_p(A:B:C)$ in definition \ref{def:delta_p_3} vanishes for $GHZ_N$-type entanglement for any $n$. Thus, $\Delta^{(3)}_p(A:B:C)$ captures tripartite entanglement except $GHZ_N$-type entanglement\footnote{Similar to the quantities proposed in \cite{Balasubramanian:2024ysu}, the zero set of $\Delta^{(3)}_p(A:B:C)$ contains the quantum states only with $GHZ_N$-type entanglement.}.

In general, the tripartite EoP $E^{(3)}_p(A:B:C)$ captures both bipartite and tripartite correlations\cite{Umemoto-2018-MultipartiteEoP}. Thus, our first step is to subtract all the possible contributions from the bipartite correlations, i.e.,
\begin{equation}\label{eq:delta3_coefficient}
        \Delta_p^{(3)}(A:B:C) := E_p^{(3)}(A:B:C)  - c^{(2)}\Big[E_p^{(2)}(A:BC) + E_p^{(2)}(AB:C) + E_p^{(2)}(AC:B) \Big]
\end{equation}
where $c^{(2)}$ is a coefficient. The simplest choice for the coefficient can be made by satisfying the following conditions so that $\Delta_p^{(3)}(A:B:C)$ is a tripartite correlation signal; (i) $\Delta^{(3)}_p(A:B:C)=0$ for quantum states only with bipartite entanglement, and (ii) $\Delta^{(3)}_p(A:B:C)\geq 0$ for any quantum states. 

Explicitly, by imposing $(i)$ to \eqref{eq:delta3_coefficient} for $\rho_{ABC}=\rho_{AB} \otimes \rho_C$ as an example, we should have
\begin{equation}
\begin{split}
    \Big(1-2c^{(2)}\Big) E^{(2)}_p(A:B) = 0.
\end{split}
\end{equation}
Hence, we have $c^{(2)}=1/2$\footnote{We found a similar analysis in \cite{Iizuka:2025caq} where the authors have used multi-entropy \cite{Gadde:2022cqi}.}. Note that, $\Delta_p^{(3)}(A:B:C)=0$ for any coefficient for a fully product state (see \eqref{eq:EoP-product}).

We can easily see that the condition, $(ii)$ $\Delta^{(3)}_p(A:B:C)\geq 0$ for any quantum states, is satisfied when $c^{(2)}=1/2$ by \eqref{eq:lowerbound_Epn}, i.e.,
\begin{equation}
    E^{(3)}_p(A:B:C)\geq \frac{1}{2}\{E^{(2)}_p(A:BC)+E^{(2)}_p(B:AC)+E^{(2)}_p(C:AB)\}.
\end{equation}

Due to the above observations, we can propose \eqref{eq:delta3_coefficient} with $c^{(2)}=1/2$ as the definition of the tripartite correlation signal. We summarize the two conditions as the properties of $\Delta_p^{(3)}(A:B:C)$ below. In addition to them, we also show that $\Delta_p^{(3)}(A:B:C)=0$ when $\rho_{ABC}$ is pure, which can be proved by applying \eqref{eq:EoP-pure}. 
\begin{definition}[Tripartite correlation signal]\label{def:delta_p_3}
    For a density matrix $\rho_{ABC}$,
    \begin{equation}\label{eq:delta_p_3}
    \begin{split}
        \Delta_p^{(3)}(A:B:C) :=& E_p^{(3)}(A:B:C)  - \frac{1}{2} \Big[E_p^{(2)}(A:BC) + E_p^{(2)}(B:AC) + E_p^{(2)}(C:AB) \Big]\\
        :=& E_p^{(3)}(A:B:C)  - \frac{1}{2} \Big[\Delta_p^{(2)}(A:BC) + \Delta_p^{(2)}(B:AC) + \Delta_p^{(2)}(C:AB) \Big].\\
    \end{split}
    \end{equation}
\end{definition}

\begin{proposition}[Properties of $\Delta^{(3)}_p(A:B:C)$]\label{prop:delta3p}$\;$\newline

    \begin{enumerate}
        \item[0.] If a density matrix $\rho_{ABC}$ is a mixed product state, e.g., $\rho_{ABC} =  \rho_{AB} \otimes \rho_{C}$,
        \begin{equation}
            \Delta_p^{(3)}(A:B:C) =0.
        \end{equation}
        \item Non-negative for all quantum states
        \begin{equation}
            \Delta_p^{(3)}(A:B:C) \geq 0.
        \end{equation}
        \item If a density matrix is a pure state,
        \begin{equation}
            \Delta_p^{(3)}(A:B:C) =0.
        \end{equation}
        
    \end{enumerate}
\end{proposition}
\begin{proof}$ $\newline

    \begin{enumerate}
    \item[0.] For $\rho_{ABC} =  \rho_{AB} \otimes \rho_{C}$, we consider some purified state of $\rho_{ABC}$,
    \begin{equation}
        \ket{\psi_{ABC}}\bra{\psi_{ABC}} = \ket{\psi_{AB\tilde{A}\tilde{B}}}\bra{\psi_{AB\tilde{A}\tilde{B}}} \otimes \ket{\phi_{C\tilde{C}}}\bra{\phi_{C\tilde{C}}},
    \end{equation}
    and their reduced states $\sigma_{A\tilde{A}}$, $\sigma_{B\tilde{B}}$, and $\sigma_{C\tilde{C}}$.
    Here, $\tilde{A},\tilde{B},\tilde{C}$ are the purifiers of $A,B,C$, respectively. We denote entropies of the reduced states of $A\tilde{A}$, $B\tilde{B}$, and $C\tilde{C}$ as $S_{A\tilde{A}}$, $S_{B\tilde{B}}$, and $S_{C\tilde{C}}$.
    The bipartite and tripartite EoPs can be computed as
    \begin{align}
        E^{(2)}_p(A:BC) & = E^{(2)}_p(B:AC) = \underset{\ket{\psi}_{ABC\tilde{A}\tilde{B}\tilde{C}}}{min}[S_{A\tilde{A}}+S_{B\tilde{B}}],\\
        E_p^{(2)}(C:AB)&=0,
    \end{align}
    and
    \begin{equation}
        E_p^{(3)}(A:B:C)=\underset{\ket{\psi}_{ABC\tilde{A}\tilde{B}\tilde{C}}}{min}[S_{A\tilde{A}}+ S_{B\tilde{B}}]
    \end{equation}
    Hence, one can simply check that $\Delta^{(3)}_p(A:B:C)=0$. 

    \item and 2. See the paragraph above definition \ref{def:delta_p_3}.
    
    \end{enumerate}
\end{proof}

\subsection{Examples: $GHZ_N$ and $W_N$}

In this subsection, we evaluate $\Delta^{(3)}_p(A:B:C)$ for $GHZ_N$ and $W_N$ states for some $N\in \mbZ_{+}$. In particular, they are
\begin{equation}\label{eq:GHZn}
    \ket{GHZ_N} := \frac{1}{\sqrt{2}}\big(\ket{0}^{\otimes N} + \ket{1}^{\otimes N}\big),
\end{equation}
and
\begin{equation}\label{eq:Wn}
    \ket{W_N} := \frac{1}{\sqrt{N}}\big(\ket{10\cdots 0} + \ket{010\cdots 0} + \cdots  + \ket{0\cdots 01} \big).
\end{equation}

We compute the upper bound and the lower bound of $E^{(3)}_p(A:B:C)$ and $E^{(2)}_p(A:BC)$\cite{terhalEoP2002,Umemoto-2018-MultipartiteEoP}. That is,
\begin{equation}\label{eq:uppr-lower-bounds-3}
\begin{split}
    & \min \{S_A + S_B+S_{AB}, S_B + S_C+S_{BC},  S_A + S_C+S_{AC} \} \\
    &\geq E^{(3)}_p(A:B:C) \\
    &\geq \max\{S_A +S_B+S_C-S_{ABC}, 2(S_A+S_B+S_C) - S_{AB}-S_{BC}-S_{CA}  \},
\end{split}
\end{equation}
and
\begin{equation}\label{eq:uppr-lower-bounds-2}
    2\min \{S_A, S_{AB}\} \geq E^{(2)}_p(A:BC) \geq I(A:B)+I(A:C).
\end{equation}

In the case of $\ket{GHZ_N}$, we have 
\begin{equation}\label{eq:GHZ_entanglement}
    S_A = S_B= S_C=S_{AB} = S_{BC} = S_{CA} = S_{ABC} = \ln 2.
\end{equation}
These entanglement entropies for $\ket{GHZ_N}$ for $N\geq 3$ do not depend on the size of subsystems $|A|,|B|,|C|$.
By inserting \eqref{eq:GHZ_entanglement} into \eqref{eq:uppr-lower-bounds-3} and \eqref{eq:uppr-lower-bounds-2}, we can find that the upper bound and the lower bound match. Thus,
\begin{equation}\label{eq:GHZ_N_EPs}
    E^{(3)}_p(A:B:C) = 3\ln 2,\; E^{(2)}_p(A:BC) = 2 \ln 2.
\end{equation}
Then it is simply $\Delta^{(3)}_p(A:B:C)=0$ for any $n$. This shows that $\Delta^{(3)}_p(A:B:C)$ cannot capture the entanglement pattern of $\ket{GHZ_N}$.

\begin{figure}
    \centering
    \includegraphics[width=1\linewidth]{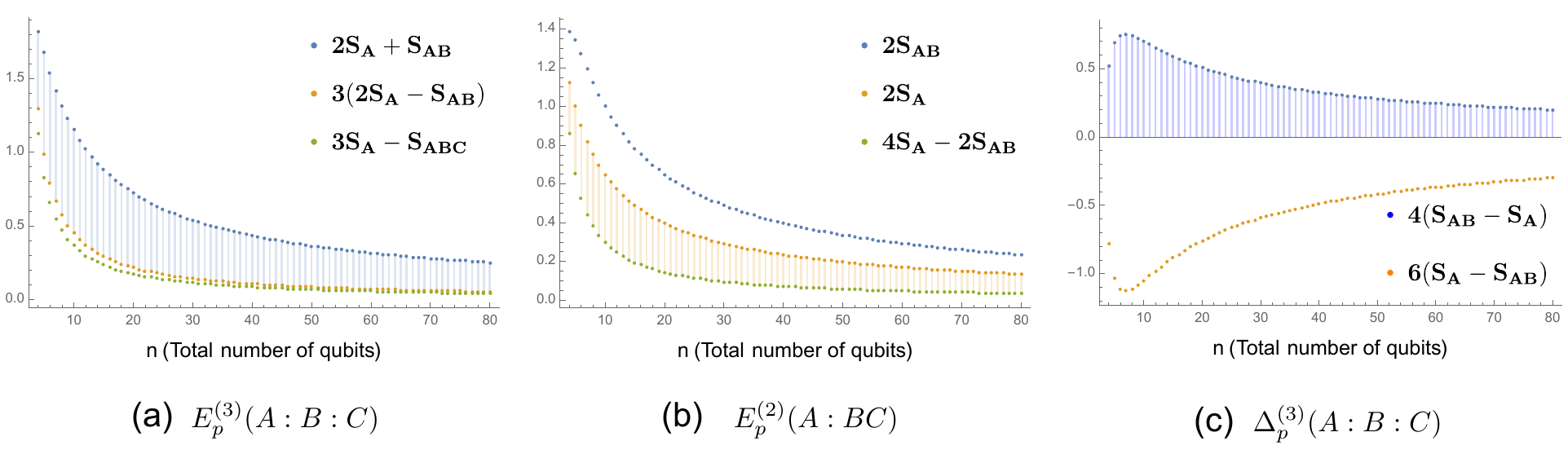}
    \caption{\small{(a) The upper bound and the candidate lower bounds in \eqref{eq:uppr-lower-bounds-3} are plotted as a function of total number $N$ of qubits for $\ket{W_N}$ \eqref{eq:Wn}. The blue dots are the upper bound. The orange and green dots are the candidate lower bounds. $E^{(3)}_p(A:B:C)$ as a function of $N$ takes a value between the blue dots and the orange dots because the orange dots maximize the lower bound. (b) Similarly, the candidate upper bounds and the lower bound of \eqref{eq:uppr-lower-bounds-2} is plotted as a function of $N$. $E^{(2)}_p(A:B:C)$ as a function of $N$ takes a value between the orange dots and the green dots because the orange dots minimize the upper bound. (c) We plotted \eqref{eq:upper-lower-bounds-delta} as a function of $n$.}}
    \label{fig:upper-lowerbounds-Wn}
\end{figure}

For $\ket{W_N}$, the entanglement entropy of the $k$-partite reduced density matrix $\rho^N_{k}$ for $k=1,\cdots, N-1$ is
\begin{equation}
    S(\rho^N_{k}) = -\frac{k}{N}\log \frac{k}{N} - \frac{N-k}{N}\log \frac{N-k}{N} 
\end{equation}
where $k$ is the number of subsystems. In general,
\begin{equation}
    S_A = S(\rho^N_{|A|}),\; S_B = S(\rho^N_{|B|}),\; S_B = S(\rho^N_{|C|}).
\end{equation} 
Here, we consider the case, such as $|A|=|B|=|C|=1$. Then, we have
\begin{equation}
\begin{split}
     S_A = S_B= S_C=&  -\frac{1}{N}\log \frac{1}{N} - \frac{N-1}{N}\log \frac{N-1}{n} ,\\  
     S_{AB} = S_{BC} = S_{CA} =&  -\frac{2}{N}\log \frac{2}{N} - \frac{N-2}{N}\log \frac{N-2}{N},\\
     S_{ABC} =&  -\frac{3}{N}\log \frac{3}{N} - \frac{N-3}{N}\log \frac{N-3}{N}.
\end{split}
\end{equation}
We can simplify \eqref{eq:uppr-lower-bounds-3} and \eqref{eq:uppr-lower-bounds-2} into
\begin{equation}\label{eq:uppr-lower-bounds-3-Wn}
    2S_A +S_{AB} \geq E^{(3)}_p(A:B:C) \geq \max\{3S_A -S_{ABC}, 3(2S_A - S_{AB})\},
\end{equation}
and
\begin{equation}\label{eq:uppr-lower-bounds-2-Wn}
    2\min \{S_A, S_{AB}\} \geq E^{(2)}_p(A:BC) \geq 4S_A -2S_{AB}.
\end{equation}
Furthermore, we have $E^{(2)}_p(A:BC)=E^{(2)}_p(B:AC)=E^{(2)}_p(C:AB)$.
Instead of giving the analytical expressions of the upper bounds and lower bounds, we plot each of them as a function of $n$ in figure \ref{fig:upper-lowerbounds-Wn}. 

The difference between the upper bound of $E^{(3)}_p(A:B:C)$ and the lower bound of $3/2 E^{(2)}_p(A:BC)$ gives a naive upper bound of $\Delta^{(3)}_p(A:B:C)$, which is $Max(\Delta^{(3)}_p(A:B:C))= 4(S_{AB} - S_A)$.
The difference between the lower bound of $E^{(3)}_p(A:B:C)$ and the upper bound of $3/2 E^{(3)}_p(A:BC)$ is negative, that is, $6(S_A - S_{AB})\leq 0$. However, $\Delta^{(3)}_p(A:B:C)\geq 0$. These simply imply that $\Delta^{(3)}_p(A:B:C)$ could saturate to $0$. 

In short, $\Delta^{(3)}_p(A:B:C)$ takes a value within the following range, i.e.,
\begin{equation}\label{eq:upper-lower-bounds-delta}
    4(S_{AB} - S_A) \geq \Delta^{(3)}_p(A:B:C) \geq 0.
\end{equation}

\section{Tripartite correlation signal $\Delta^{(3)}_w$ for holographic states}\label{sec:holography}

For a bipartite holographic state in AdS$_{d+1}$/CFT$_d$, $E_p(A:B) = E_w(A:B)$\cite{Umemoto2018EoPholography}, where $E_p(A:B)$ is the entanglement of purification and $E_w(A:B)$ is the entanglement wedge cross section.

Reference \cite{Umemoto-2018-MultipartiteEoP} proposed the multipartite entanglement wedge cross-section(MEWCS) defined below as a candidate holographic dual of MEoP. Before providing the definition, we set the notation\footnote{Our notation closely follows \cite{Umemoto2018EoPholography,Umemoto-2018-MultipartiteEoP} with minor differences.}.

Consider a constant time slice $M$ of AdS$_3$/CFT$_2$ and $n$ boundary subregions $A_1 \cup \cdots \cup A_n \subseteq \partial M$ that are basically disjoint from each other\footnote{We will also consider a case where some subregions are touching to each other.}. The entanglement wedge $M_{A_1\cdots A_n}$ on the constant time slice $M$ is defined by the bulk region surrounded by the boundary subregions $A_1 \cup \cdots \cup A_n$ and the codimension-two minimal surface $\Gamma^{(n)}_{min}(A_1:\cdots :A_n)$, i.e.,
\begin{equation}
    \partial M_{A_1\cdots A_n} = A_1\cup \cdots \cup A_n \cup \Gamma^{(n)}_{min}(A_1:\cdots :A_n).
\end{equation}

To define a cross section of $M_{A_1\cdots A_n}$, we partition $\partial  M_{A_1\cdots A_n}$ arbitrarily,
\begin{equation}
    \partial M_{A_1\cdots A_n} = \tilde{A}_1\cup \cdots \cup \tilde{A}_n,
\end{equation}
such that\footnote{We sometimes denote the subregions by $\hat{A}_m$ instead of $\tilde{A}_m$ to avoid confusion when we have two different purifications on a single geometry. See figure \ref{fig:delta3-ewcs} for instance.} 
\begin{equation}
    \tilde{A}_m \supseteq A_m, \; \forall m=1,\cdots, n.
\end{equation}

Then, the cross section $\Sigma^{(n)}_{min}(\tilde{A}_1:\cdots:\tilde{A}_n)$ is defined by
\begin{equation}
    \Sigma^{(n)}_{min}(\tilde{A}_1:\cdots:\tilde{A}_n) = \Sigma_{min}(\tilde{A}_1)\cup\cdots\cup \Sigma_{min}(\tilde{A}_n)
\end{equation}
where $\Sigma_{min}(\tilde{A}_m)$ are the minimal surface homologous to $\tilde{A}_m$ and $\partial (\tilde{A}_m)= \partial(\Sigma_{min}(\tilde{A}_m)) \in \partial M_{A_1\cdots A_n}$ for all $m=1,\cdots,n$.

The definition of the entanglement wedge cross section requires a further minimization by tuning the size of $\tilde{A}_1,\cdots, \tilde{A}_n$.
\begin{definition}[Multipartite entanglement wedge cross section\cite{Umemoto-2018-MultipartiteEoP}\footnote{An alternative definition of multipartite EWCS was given in \cite{Bao-2019-conditionalEoP} which differ from \cite{Umemoto-2018-MultipartiteEoP} by a factor of $1/n$. We adopt the definition from \cite{Umemoto-2018-MultipartiteEoP}.}]\label{def:E_W_N}
     \begin{equation} \label{eq:E_W_N}
         E_w^{(n)}(A_1:\cdots:A_n) := \underset{ A_1\subseteq \tilde{A}_1,\cdots, A_n\subseteq\tilde{A}_n}{min}\frac{L_{AdS} }{4G_N} Area\Big(\Sigma_{min}^{(n)}(\tilde{A}_1:\cdots:\tilde{A}_n )\Big)
     \end{equation}
\end{definition}

Similar to \eqref{eq:Ep2=2EoP}, the bipartite EWCS $E^{(2)}_w(A:B)$ is defined to be twice the EWCS $E_w(A:B)$, i.e.,
\begin{equation}
    E^{(2)}_w(A_1:A_2) := 2E_w(A:B).
\end{equation}

\subsection{Definition and properties}

Inspired by definition \ref{def:delta_p_3}, we now introduce a definition of a tripartite correlation signal for holographic states.
\begin{definition}[Tripartite correlation signal for a holographic state] 

For a holographic state $\ket{\psi}_{ABCO}$
\begin{equation}\label{eq:delta_w_3}
\begin{split}
    \Delta_w^{(3)}(A:B:C) & := E^{(3)}_w(A:B:C)  - \frac{1}{2} \Big[ E_w^{(2)}(A:BC) +E_w^{(2)}(AB:C) + E_w^{(2)}(AC:B) \Big]\\
    & = E^{(3)}_w(A:B:C)  - \Big[ E_w(A:BC) +E_w(AB:C) + E_w(AC:B) \Big]\\
\end{split}
\end{equation}
\end{definition}

Similar to the conjecture $E_w(A:B) = E_p (A:B)$ in the bipartite case, \cite{Umemoto-2018-MultipartiteEoP} proposed the following conjecture. It was further explored in \cite{Bao-2019-conditionalEoP}.
\begin{conjecture}[\cite{Umemoto-2018-MultipartiteEoP,Bao-2019-conditionalEoP}]\label{con:ewcs-conjecture}
    \begin{equation}
         E_p^{(n)}(A_1: \cdots : A_n)  = E_w^{(n)}(A_1: \cdots : A_n)
    \end{equation}
\end{conjecture}
By conjecture \ref{con:ewcs-conjecture}, we have 
\begin{equation}
    \Delta_p^{(n)} = \Delta_w^{(n)}.
\end{equation}
Our studies of $\Delta^{(3)}_w$ for holographic states rest on the above conjecture. That is, if the conjecture holds,  $\Delta^{(3)}_w$ must satisfy the properties of $\Delta^{(3)}_p$ in proposition \ref{prop:delta3p}. We exhibit those properties from the viewpoints of bulk geometries.

\begin{proposition}[Properties of $\Delta^{(3)}_w(A:B:C)$]\label{prop:delta_p_3_properties}$\;$\newline

    \begin{enumerate}
    \item[0.] If a holographic state $\ket{\psi}_{ABCO}$ does not have a connceted tripartite entanglement wedge $M_{ABC}$,
        \begin{equation}
            \Delta_w^{(3)}(A:B:C) =0.
        \end{equation}
    \item Non-negative up to the leading order of $G_N$ for all holographic states 
        \begin{equation}
            \Delta_w^{(3)}(A:B:C) \geq 0.
        \end{equation}
    \item If a density matrix is a pure state,
        \begin{equation}
            \Delta_w^{(3)}(A:B:C) =0.
        \end{equation}
    \end{enumerate}
    
\end{proposition}

\begin{proof}
Consider subregions $A,B,C,O$ on the boundary of AdS$_3$/CFT$_2$ in a fully connected phase, see figure \ref{fig:delta3-ewcs}. Let us set $L_{AdS}/4G_N =1$. Here, $O$ is a complementary boundary region of $A\cup B\cup C$. In the fully connected phase, the complementary region $O$ is partitioned into three regions, i.e., $O = O_1 \cup O_2 \cup O_3$. For a reduced density matrix $\rho_{ABC}$ of the boundary subsystem $ABC$, the bulk dual geometry of a purified state is the entanglement wedge $M_{ABC}$ whose boundary is 
\begin{equation}
    \partial M_{ABC} = A\cup B\cup C \cup \Gamma^{(3)}_{min}(A:B:C).
\end{equation}
By following the definition of MEoP, we have 
\begin{equation}
    E_w^{(2)}(A:BC) = 2l_a,\; E_w^{(2)}(B:AC) = 2l_b,\; E_w^{(2)}(C:AB) = 2l_c.
\end{equation}
and
\begin{equation}
    E_w^{(3)}(A:B:C) = L_a + L_b +L_c.
\end{equation}
For simplicity, we assume that $\tilde{A},\tilde{B},\tilde{C}$ for bipartite EoP and $\hat{A},\hat{B},\hat{C}$ for tripartite EoP are all optimized and give $l_a,l_b,l_c$ and $L_a,L_b,L_c$ as, figure \ref{fig:delta3-ewcs},
\begin{equation}
    l_a: = Area\Big(\Sigma_{min}(\tilde{A})\Big),\;l_b: = Area\Big(\Sigma_{min}(\tilde{B})\Big),\;l_c: = Area\Big(\Sigma_{min}(\tilde{C})\Big),\;
\end{equation}
and 
\begin{equation}
    L_a: = Area\Big(\Sigma_{min}(\hat{A})\Big),\;L_b: = Area\Big(\Sigma_{min}(\hat{B})\Big),\;L_c: = Area\Big(\Sigma_{min}(\hat{C})\Big).
\end{equation}

\begin{figure}
    \centering
    \includegraphics[width=0.9\linewidth]{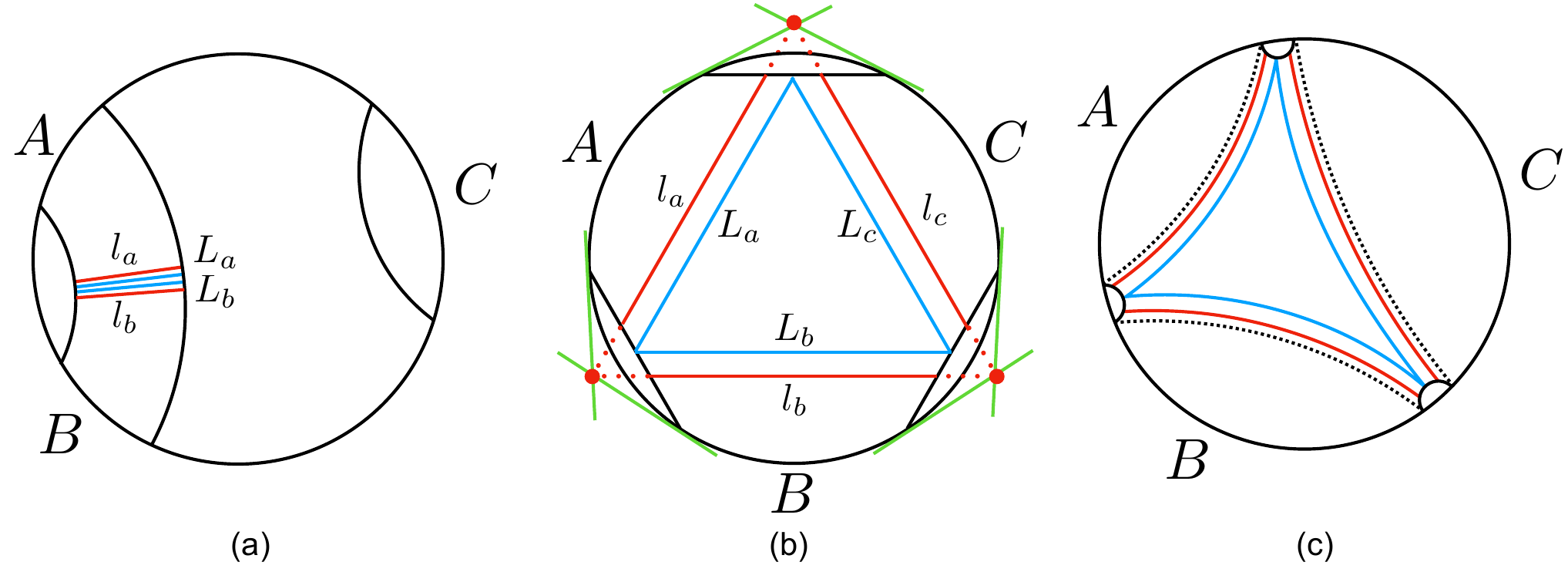}
    \caption{\small{(a) The red curves are $l_a$ and $l_b$. The blue curves are $L_a$ and $L_b$. The figure is exaggerated so that one can see the explicit contributions of $L_a,L_b,l_a,l_b$ to the corresponding MEoP. See \eqref{eq:bipartite_connected_phase} and \eqref{eq:bipartite_connected_phase_EWCS}. (b) A graphical proof using the Beltrami-Klein model. Black lines are the minimal surfaces of the entanglement wedge $\Gamma^{(3)}_{min}(A:B:C)$ in a fully connected phase. Green lines are tangent to the boundary of the boundary subregions $A,B,C$. Red points represent the intersections of the green lines. $L_a,L_b,L_c$ form a triangle, whereas $l_a,l_b,l_c$ cannot form a triangle within the disk unless $ABC$ covers the whole boundary of the disk. (c) The exaggerated figure representing the phase just right before ABC covers the whole boundary. Black dotted curves are the RT surfaces $\Gamma_{min}(A),\Gamma_{min}(B),\Gamma_{min}(C)$ of boundary subregion $A,B,C$. When $ABC$ becomes the whole boundary, we have \eqref{eq:pure_state_ewcs}. }}
    \label{fig:delta-w-proofs}
\end{figure}

$ $\newline
\noindent
\begin{enumerate}
    \item[0.] (Figure \ref{fig:delta-w-proofs} (a)) When the state has only bipartite entanglement, such as $\rho_{ABC} = \rho_{AB}\otimes \rho_C$, the entanglement wedge $AB$ is in the connected phase, but the entanglement wedges $AB$ and $C$ are disconnected. After the phase transition from the fully connected phase to such a phase, we have, up to $O(1/G_N)$,
    \begin{equation}\label{eq:bipartite_connected_phase}
        E_w^{(3)}(A:B:C) = E_w^{(2)}(A:BC) = E_w^{(2)}(B:AC), \;  E_w^{(2)}(C:AB) = 0
    \end{equation}
    because
    \begin{equation}\label{eq:bipartite_connected_phase_EWCS}
        L_a = L_b = l_a=l_b, \; L_c=l_c = 0.
    \end{equation}
    Inserting \eqref{eq:bipartite_connected_phase} to \eqref{eq:delta_w_3}, we get
    \begin{equation}
        \Delta^{(3)}_{w}(A:B:C) = 0 . 
    \end{equation}

    \item[1.] (Figure \ref{fig:delta3-ewcs} and \ref{fig:delta-w-proofs} (b)) Consider a hyperbolic triangle with angles $\theta_1,\theta_2,\theta_3$. The sum of the angles is strictly smaller than $\pi$, i.e.,
    \begin{equation}
        \pi > \theta_1+\theta_2+\theta_3.
    \end{equation}
    Moreover, if the vertices of the triangles are in the interior of the Poincaré disk and thus its edges have a finite length, then, the angles are strictly larger than $0$, that is,
    \begin{equation}\label{eq:finitetriangle_angles}
        \theta_1,\theta_2,\theta_3>0.
    \end{equation}

    For $l_a,l_b,l_c$ to form a triangle with the angles $\theta_1,\theta_2,\theta_3$, all of the angles become
    \begin{equation}
        \theta_m=0,\; m=1,2,3
    \end{equation}
    since they are the unique geodesics that are perpendicular to the RT surfaces where their endpoints land. However, this contradicts with \eqref{eq:finitetriangle_angles}. Hence, $l_a,l_b,l_c$ do not share their endpoints unless their endpoints are asymptotically close to the boundary so that they are pairwise parallel asymptotically.

    On the contrary, $L_a,L_b,L_c$ can form a hyperbolic triangle with, for instance, angles $\theta_1,\theta_2,\theta_3$ of a finite area by construction, figure \ref{fig:delta3-ewcs}. From the above discussion and the fact that $l_a,l_b,l_c$ are i) the unique shortest geodesics between the RT surfaces and ii) common perpendicular to the RT surfaces, at least two of $L_i$'s are always larger than $l_i$'s, for instance,
    \begin{equation}
        L_a=l_a,\; L_b>l_b,\; L_c>l_c.
    \end{equation}
    Therefore, $\Delta^{(3)}_w\geq 0$. For a graphical proof based on the Beltrami-Klein model in AdS$_3$/CFT$_2$, see figure \ref{fig:delta-w-proofs} (b).

    \item[2.] (Figure \ref{fig:delta-w-proofs} (c)) If a holographic state for $ABC$ is a pure state, we have
    \begin{equation}\label{eq:pure_state_ewcs}
        L_a = l_a=\Gamma_{min}(A) , \;L_b = l_b=\Gamma_{min}(B),\; L_c=l_c=\Gamma_{min}(C).
    \end{equation}
    Thus, $\Delta_w^{(3)}(A:B:C)=0$.
\end{enumerate}

\end{proof}

\begin{figure}[t]
    \centering
    \includegraphics[width=0.4\linewidth]{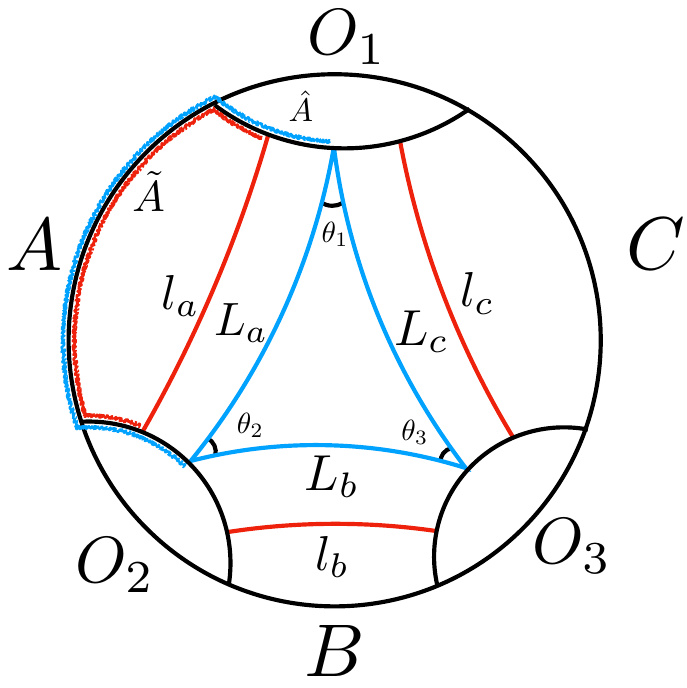}
    \caption{\small{The connected entanglement wedge of $ABC$ in a static time-slice of $AdS_3/CFT_2$.
    The red curves denote $E_w^{(2)}(A:BC) = 2E_w(A:BC)=2l_a$ and permutations, $E_w^{(2)}(B:AC)=2E_w(B:AC)=2l_b$ and $E_w^{(2)}(C:AB)=2E_w(C:AB)=2l_c$. The blue curve denotes $E_w^{(3)}(A:B:C)=L_a+L_b+L_c$, which form a hyperbolic triangle. $\theta_1,\theta_2,\theta_3$ are the angles of the hyperbolic triangle. $\tilde{A}\supset A$ and $\hat{A}\supset A$ are the optimal boundary subregions to compute $l_a$ and $L_a$. $\tilde{B},\hat{B}$ and $\tilde{C},\hat{C}$ are defined similarly, although they are suppressed from the figure. }}
    \label{fig:delta3-ewcs}
\end{figure}

\subsection{Example: Pure $AdS_3$}
Consider the vacuum $AdS_3$ and take three equal subregions $A$, $B$ and $C$ symmetrically chosen. We have drawn the surfaces relevant for multipartite entanglement of purification in figure \ref{fig:delta3-ewcs}. Let $\phi_A, \phi_B, \phi_C$ denote the mid-points of the boundary subregions $A,B,C$ respectively, which are given as follows:
\begin{equation}
    A: [\phi_A-\alpha_A, \phi_A+\alpha_A], \quad B: [\phi_B-\alpha_B, \phi_B+\alpha_B], \quad C: [\phi_C-\alpha_C, \phi_C+\alpha_C].
\end{equation}

For our symmetric choice of regions, we have
\begin{equation}\label{eq:regions_ABC}
\begin{split}
    & \phi_A=0, \quad \phi_B=2\pi/3 , \quad \phi_C=4\pi/3 \\
    & \alpha_A=\alpha_B=\alpha_C=\alpha
\end{split}
\end{equation}

For the fully connected phase given in figure \ref{fig:delta3-ewcs}, we have
\begin{equation}\label{eq:E2ABC}
    E^{(2)}_w(A:BC)=\frac{2L_{AdS}}{4G} \log\left(
  \frac{
    \left( \sqrt{\sin\left(\frac{2\pi}{3}\right) \sin\left(\frac{\pi}{3}\right)} + \sqrt{\sin(\alpha) \sin\left(\frac{\pi}{3} + \alpha\right)} \right)^2
  }{
    \sin\left(\frac{2\pi}{3} + \alpha\right) \sin\left(\frac{\pi}{3} - \alpha\right)
  }
\right),
\end{equation}
and similarly for $E(B:AC)$ and $E(C:AB)$. Here, we have used the Beltrami-Klein model for this computation following \cite{Nguyen:2017yqw}. We also have,
\begin{equation}\label{eq:E3ABC}
    E^{(3)}_w(A:B:C)=\frac{3L_{AdS}}{4G} \operatorname{arccosh}\left( 
  1 + \frac{6\tan^2{\left(\frac{\pi}{4}-\frac{\alpha}{2}\right)}}{\left(1-\tan^2{\left(\frac{\pi}{4}-\frac{\alpha}{2}\right)}\right)^2} 
\right),
\end{equation}
where we have identified the minimizing configuration of $\hat{A}, \hat{B}, \hat{C}$ to be the symmetric one, and computed the hyperbolic lengths for the same on a Poincare disk. We see that as $\alpha\rightarrow \frac{\pi}{3}$, both of these surfaces approach the RT surface, and are divergent. It is straightforward to check that both these expressions \ref{eq:E2ABC} and \ref{eq:E3ABC} have the same leading order divergent behavior near $\alpha=\frac{\pi}{3}-\delta$ (where $\delta>0$ is arbitrarily small), captured by (see figure \ref{fig:Delta3vsalpha_comparisons})
\begin{equation}\label{eq:E2-limit-E3}
     \lim_{\alpha \rightarrow \frac{\pi}{3}-\delta} \frac{1}{3} E_w^{(3)}(A:B:C) \sim \log{\left(\frac{\sqrt{3}}{\delta}\right)} \sim \lim_{\alpha \rightarrow \frac{\pi}{3}-\delta} \frac{1}{2}E_w^{(2)}(A:BC),
\end{equation}

\begin{figure}[t]
    \centering
    \includegraphics[width=1\linewidth]{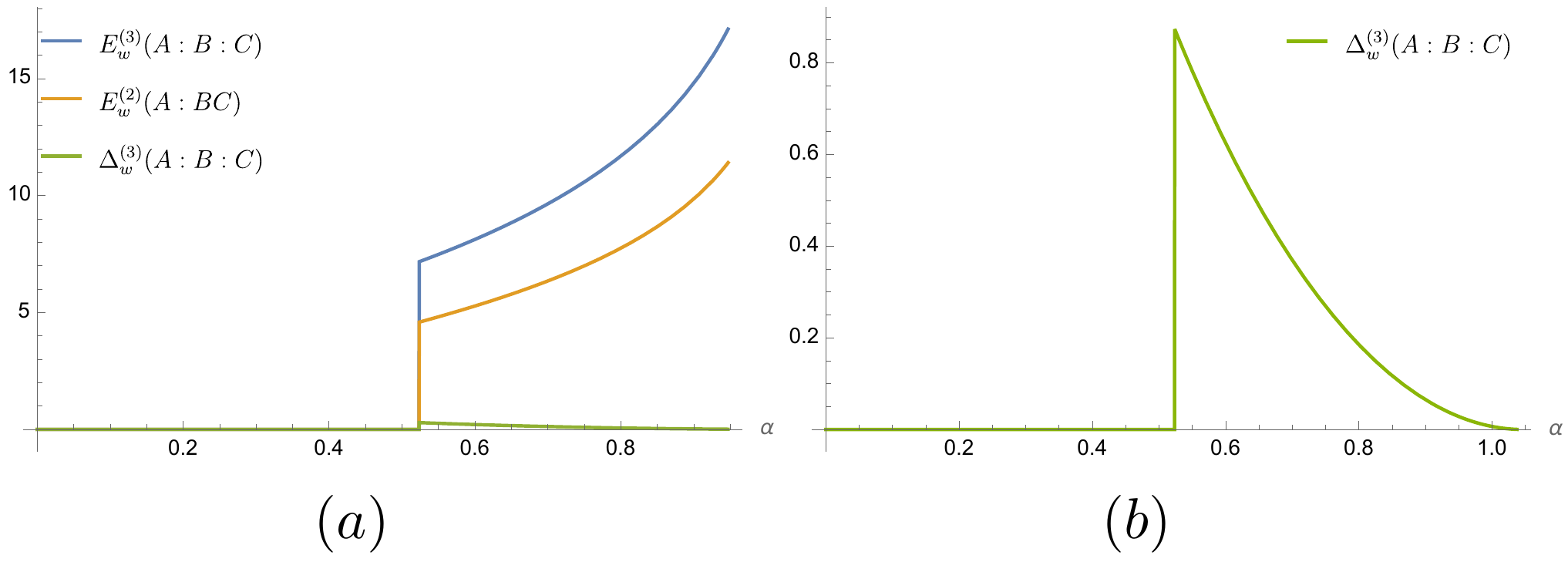}
    \caption{\small{We have set $L_{AdS}/4_{G_N}=1$. At $\alpha=\frac{\pi}{6}$, there is a discontinuous jump in $\Delta_w^{(3)}(A:B:C)$, due to the disconnected–connected phase transition of the tripartite entanglement wedge $M_{ABC}$, after which $\Delta_w^{(3)}(A:B:C)$ monotonically decreases. (a) A comparison of $E_w^{(3)}(A:B:C)$, $E_w^{(2)}(A:BC)$ and $\Delta_w^{(3)}(A:B:C)$. The $x$-axis is cut-off at $\alpha=\pi/3-10^{-1}$. (b) $\Delta_w^{(3)}(A:B:C)$ as a function of $\alpha$. The $x$-axis is cut-off at $\alpha=\pi/3-10^{-2}$. }}
    \label{fig:Delta3vsalpha_comparisons}
\end{figure}

Let us discuss the phases of this configuration as we vary the size of the boundary regions $A$, $B$ and $C$, which for this symmetric case, we have parameterized by $\alpha$,

\begin{itemize}
    \item When $\alpha$ is small enough, i.e., $\alpha< \alpha_{0}$ (we found $\alpha_{0}\approx \frac{1}{2}\left(\frac{\pi}{3}\right)$), where $\alpha_{0}$ is the value of $\alpha$ at which connected-disconnected phase transition of the entanglement wedge $M_{ABC}$ of $ABC$ takes place, we have the disconnected phase. In this phase, $E_w^{(3)}(A:B:C)=0$ and also $E_w^{(2)}(A:BC)=E_w^{(2)}(B:AC)=E_w^{(2)}(C:AB)=0$. This gives us $\Delta_w^{(3)}(A:B:C)=0$. This corresponds to a completely factorizable state $\ket{\psi}_{\hat{A}\hat{B}\hat{C}}\equiv\ket{\psi}_{\hat{A}}\otimes\ket{\psi}_{\hat{B}}\otimes\ket{\psi}_{\hat{C}}$.

    \item In the connected phase, where $\alpha_{0}< \alpha <\pi/3$, we have a monotonically decreasing $\Delta_w^{(3)}>0$. More precisely, at $\alpha=\alpha_{0}$, there is a discontinuous jump in $\Delta_w^{(3)}$ due to change of topology, and then it continuously approaches $0$ at $\alpha=\pi/3$.

    \item When $\alpha=\pi/3$, then the blue curve and the red curve approaches to coincide with each other and their areas are approximated by equation \ref{eq:E2-limit-E3}. Therefore, $\Delta_w^{(3)}=0$. This corresponds to a pure state $\ket{\psi}_{ABC}$ defined on the full boundary region. It is noteworthy that the individual quantities $E_w^{(3)}(A:B:C)$ and $E_w^{(2)}(A:BC)$ (and permutations) are divergent by themselves (and we have to use a regulator at the boundary), but $\Delta_w^{(3)}$ built from their linear combination is well-behaved, and all divergences get cancelled (see figure \ref{fig:Delta3vsalpha_comparisons}).
\end{itemize}

Before we move on to the next section, we discuss i) the monotonicity of $\Delta^{(3)}_w(A:B:C)$, ii) a phase of geometry with genuine multipartite entanglement, and iii) the polygamy of MEoP in our specific setup where $|A|=|B|=|C|$ and the distance between the subregions are equal.

In figure \ref{fig:Delta3vsalpha_comparisons}, $\Delta^{(3)}_w$ monotonically increases as $\alpha$ decreases, i.e.,
\begin{equation}
    \Delta^{(3)}_w(A_1:B_1:C_1)\leq \Delta^{(3)}_w(A_2:B_2:C_2)
\end{equation}
for $A_1>A_2, B_1>B_2, C_1>C_2$. In other words, $\Delta^{(3)}_w$ monotonically increases under partial tracing each subregion evenly. We could redefine the measure by
\begin{equation}\label{eq:tilde_delta_w_3}
    \tilde{\Delta}^{(3)}_w(A:B:C):= \max [ \tilde{\Delta}^{(3)}_w(A:B:C)] - \tilde{\Delta}^{(3)}_w(A:B:C)
\end{equation}
so that
\begin{equation}
    \tilde{\Delta}^{(3)}_w(A_1:B_1:C_1)\geq \tilde{\Delta}^{(3)}_w(A_2:B_2:C_2)
\end{equation}
for $A_1>A_2, B_1>B_2, C_1>C_2$. We would like to ask what its operational interpretation is in a future work.  

In figure \ref{fig:ewcs-tri-polygamy} (a), we plotted $E^{(2)}_w(A:BC)$ and $E^{(2)}_w(A:B)$. On one hand, $E^{(2)}_w(A:BC)$ is computed with the tripartite entanglement wedge $M_{ABC}$ in a fully connected phase. On the other hand, $E^{(2)}_w(A:B)$ is computed with the bipartite entanglement wedge $M_{AB}$ in a fully connected phase. We observe that the transition point of the tripartite entanglement wedge $M_{ABC}$ occurs earlier than that of the bipartite entanglement wedge $M_{AB}$. As expected, the tripartite entanglement wedge $M_{ABC}$ exhibits genuine multipartite entanglement among $A,B,C$. We can explicitly see it from figure \ref{fig:ewcs-tri-polygamy} that the genuine tripartite entanglement among $A,B,C$ is captured by $\Delta_w^{(3)}(A:B:C)$, and, thus, is present.
We will briefly discuss its potential connection to multipartite entanglement distillation and holographic quantum error correction in section \ref{sec:discussions}.

\begin{figure}[t]
    \centering
    \includegraphics[width=1\linewidth]{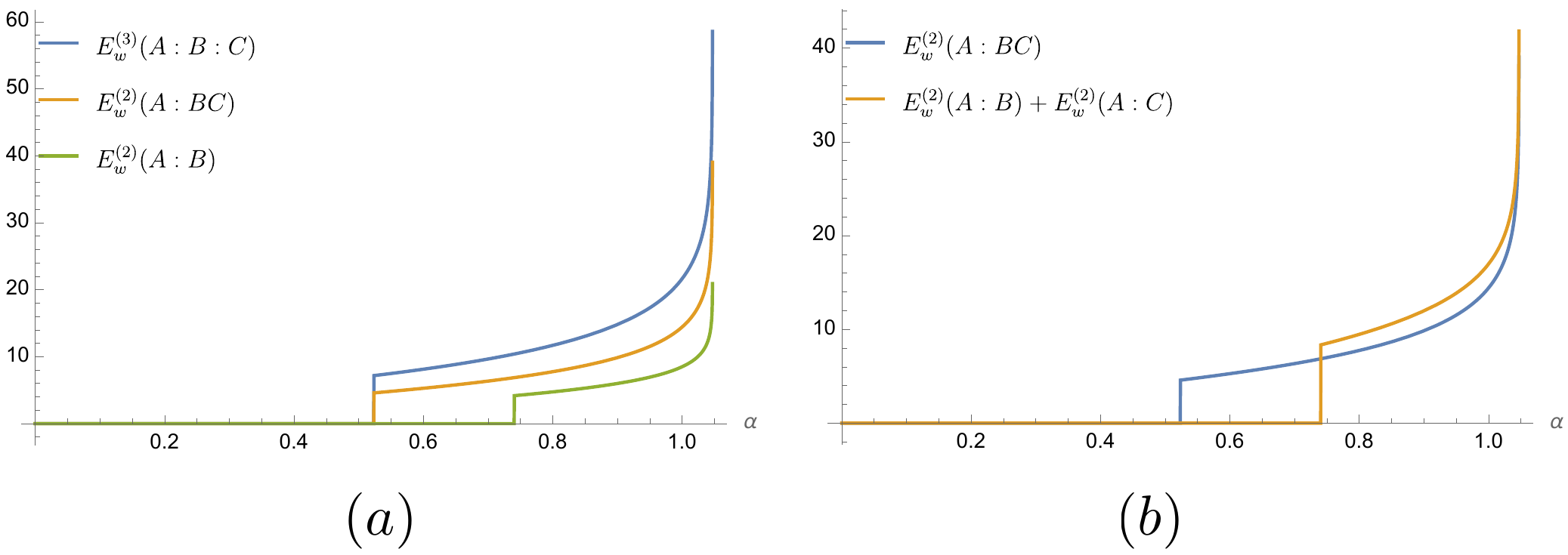}
    \caption{\small{We have set $L_{AdS}/4G_N=1$. The transition points of the blue curves happen at $\alpha^{(3)}=\pi/6$, and are the same as in figure \ref{fig:Delta3vsalpha_comparisons}. The transition points of the orange curves happen at $\alpha^{(2)}=\pi/3\sqrt{2}$. The $x$-axis is cut-off at $\alpha=\pi/3-10^{-4}$. (a) The phase $\alpha^{(2)}\leq \alpha \leq \alpha^{(3)}$ exhibits the genuine tripartite entanglement. (b) The polygamy of $E^{(2)}(A:BC)$.}}
    \label{fig:ewcs-tri-polygamy}
\end{figure}

The polygamy of $E^{(2)}_p(A:BC)$ or $E_p(A:BC)$ has been proven to hold for a pure state \eqref{eq:polygamy-MEoP}. In figure \ref{fig:ewcs-tri-polygamy} (b), we can see that the polygamy of holographic MEoP, 
\begin{equation}
    E_w^{(2)}(A:BC) \leq E_w^{(2)}(A:B) + E_w^{(2)}(A:C),
\end{equation}
does not hold for a mixed state, but does hold for a pure state at $\alpha \to \frac{\pi}{3}$ in our setup. We have the polygamy when $\alpha$ is larger than the transition point of $E_w^{(2)}(A:B)+E_w^{(2)}(A:C)$, or that of the bipartite entanglement wedge $M_{AB}$, whereas we don't when $\alpha$ is smaller than the transition point. We should note that this is still compatible with the monogamy of mutual information(MMI)\cite{Hayden2013MMI}
\begin{equation}
    I(A:BC)\geq I(A:B) + I(A:C)
\end{equation}
since the phase where $M_{ABC}$ is connected, but $M_{AB}$ is disconnected, has $I(A:BC)\neq 0$ and $I(A:B)=I(A:C) = 0$. We will leave the explorations of these quantities and the Markov gap\cite{Hayden:2021gno} for the purpose of understanding the multipartite entanglement pattern in holographic geometry for the future.

\section{Towards $n$-partite correlation signal $\Delta_p^{(n)}$}\label{sec:delta_n}


In this section, we provide the definition of $\Delta^{(4)}_p(A:B:C:D)$ and briefly discuss its properties listed in proposition \ref{prop:delta4p}. We found that the quantity is sign indefinite because $\Delta^{(4)}_p(A:B:C:D)$ is proportional to tripartite information $I_3(A:B:C)$\cite{} for a four-partite pure state. For example, it becomes negative for four-partite holographic pure states, whereas it is positive for $\Delta^{(4)}_p(A:B:C:D)\geq 0$ for GHZ$_N$ states for any $N\geq 4$, see appendix \ref{app:delta4p}.

We define $\Delta^{(4)}_p(A:B:C:D)$ as follows.
\begin{definition}[A candidate four-partite correlation signal for a density matrix]\label{def:delta_p_4}
    For a mixed density matrix $\rho_{ABCD}$, $\Delta^{(4)}_p(A:B:C:D)$ is defined as
    \begin{equation}\label{eq:delta_p_4}
    \begin{split}
        \Delta^{(4)}_p&(A:B:C:D) \\
        :=& E^{(4)}_p(A:B:C:D)\\
        &\quad - \frac{1}{3} \{E^{(3)}_p(AB:C:D) +E^{(3)}_p(AC:B:D)+E^{(3)}_p(AD:B:C) \\
        &\qquad \qquad+E^{(3)}_p(BC:A:D)+E^{(3)}_p(BD:A:C)+E^{(3)}_p(CD:A:B)\} \\
        &\quad + \frac{1}{6}\{E^{(2)}_p(A:BCD) + E^{(2)}_p(B:ACD) + E^{(2)}_p(C:ABD) + E^{(2)}_p(D:ABC)\}\\
        &\quad + \frac{1}{6}\{E^{(2)}_p(AB:CD) + E^{(2)}_p(AC:BD) +E^{(2)}_p(AD:BC)\} \\
        :=& E^{(4)}_p(A:B:C:D) \\
        &\quad - \frac{1}{3} \{\Delta^{(3)}_p(AB:C:D) +\Delta^{(3)}_p(AC:B:D)+\Delta^{(3)}_p(AD:B:C) \\
        &\qquad \qquad+\Delta^{(3)}_p(BC:A:D)+\Delta^{(3)}_p(BD:A:C)+\Delta^{(3)}_p(CD:A:B)\} \\
        &\quad - \frac{1}{3}\{\Delta^{(2)}_p(A:BCD) + \Delta^{(2)}_p(B:ACD) + \Delta^{(2)}_p(C:ABD) + \Delta^{(2)}_p(D:ABC)\}\\
        &\quad - \frac{1}{6}\{\Delta^{(2)}_p(AB:CD) + \Delta^{(2)}_p(AC:BD) +\Delta^{(2)}_p(AD:BC)\}.\\
    \end{split}
    \end{equation}
    
\end{definition}

The coefficients in front of each term is determined by imposing $\Delta^{(4)}_p(A:B:C:D)=0$ for quantum states only with bipartite and tripartite entanglement. A similar logic was applied to find the coefficients of $\Delta^{(3)}_p(A:B:C)$ in section \ref{sec:density-matrix}. 

To make the discussion parallel to the case of $\Delta_p^{(3)}(A:B:C)$, we study the following properties of $\Delta_p^{(4)}(A:B:C:D)$ in appendix \ref{app:delta4p}.
\begin{proposition}[Properties of $\Delta^{(4)}_p(A:B:C:D)$]\label{prop:delta4p}$\;$\newline

    \begin{enumerate}
        \item[0.] If a density matrix $\rho_{ABCD}$ is a mixed product state, e.g., $\rho_{ABCD} =  \rho_{ABC} \otimes \rho_{D}$,
        \begin{equation}
            \Delta_p^{(4)}(A:B:C:D) =0.
        \end{equation}
        \item Non-negative for GHZ$_N$ states for any $N\in \mbZ_+$, i.e.,
        \begin{equation}
            \Delta_p^{(4)}(A:B:C:D) \geq 0.
        \end{equation}
        \item If a density matrix is a four-partite pure state,
        \begin{equation}\label{eq:delta_p4_pure_state}
        \begin{split}
        \Delta_p^{(4)}(A:B:C:D) &= \frac{1}{3} \{S_A+S_B+S_C+S_D- (S_{AB}+S_{AC}+S_{AD})\} \\
        &= \frac{1}{3} I_3(X:Y:Z),
        \end{split}
        \end{equation}
        where 
        \begin{equation}
            I_3(X:Y:Z) = I(X:Y) + I(X:Z)-I(X:YZ),
        \end{equation}
        and $X,Y,Z = \{A,B,C,D$\}.
        
    \end{enumerate}
\end{proposition}
\begin{proof}
    See appendix \ref{app:delta4p}.
\end{proof}




We postpone any further explorations of $\Delta^{(n)}_p$ and $\Delta^{(n)}_w$ for quantum states, including holographic ones. It is also in our future interest to develop a systematic method for determining their coefficients.

\section{Discussions}\label{sec:discussions}

\subsection{Markov gap $S_R-I$ and tripartite correlation signal $\Delta^{(3)}_w$}

It was conjectured in \cite{Cui2020BitThreads} that a holographic state contains mostly bipartite entanglement, which was later called the Mostly Bipartite Conjecture(MBC) in \cite{Akers2020EWCSTripartite}. In \cite{Akers2020EWCSTripartite}, the MBC was found to be incompatible with the conjectures i) $S_R(A:B) = 2E_w(A:B)$\cite{Dutta:2019gen}, and ii) $E_p^{(2)}(A:B) = 2E_w(A:B)$, or $E_p(A:B) = EW(A:B)$\cite{Umemoto2018EoPholography}. It was argued that this incompatibility is attributed to the contributions of tripartite entanglement to the entanglement wedge cross-sections $E_w(A:B)$. This aspect was further studied and observed as the universal lower bound of \textit{Markov gap}, $S_R(A:B)-I(A:B)$ in \cite{Hayden:2021gno}, which indicates that the boundaries in the EWCS should capture tripartite entanglement.

We observed that our quantity $\Delta^{(3)}_w$ also captures some types of tripartite entanglement. This observation becomes prominent in the phase where $E^{(2)}_w(A:BC)\neq 0 $ and $E^{(2)}_w(A:B) = 0$. This phase does not contain bipartite correlations between boundary subregion $A$ and $B$, $A$ and $C$, and $B$ and $C$. Although there are still bipartite correlations between subregion $A$ and $BC$, $B$ and $AC$, and $C$ and $AB$, those contributions are subtracted when defining $\Delta^{(3)}_w(A:B:C)$. Moreover, we proved $\Delta^{(3)}_p(A:B:C)=0$ for any $GHZ$-type entanglement. Together with the conjecture $\Delta^{(3)}_w(A:B:C)=\Delta^{(3)}_p(A:B:C)$, it implies at least that there is a contribution of non-GHZ type tripartite entanglement. It is interesting to explore which geometric objects can capture which and what types of multipartite entanglement. 

An approach toward the questions would be to study operational interpretations of $\Delta^{(3)}_w(A:B:C)$ or $\tilde{\Delta}^{(3)}_w(A:B:C)$ in \eqref{eq:tilde_delta_w_3} in terms of, for instance, distillation of multipartite entanglement\cite{Mori2022distillation,Mori:2024gwe,Gun2024distillation}, or holographic quantum error correction codes\cite{Almheiri2015bulklocality,Harlow2017RTfomulaQEC,Pastawski2015Happy,Pastawski:2015qua,Faulkner:2020hzi,Furuya2022Petz,Chandra2023RTN-HoloCodes,Verlinde2013BH-QEC,Kar2023Non-isometric,Leutheusser2025subsubduality,Czech:2025jnw,Faulkner:2022ada,Freivogel2016gaugeQECadscft}.

\subsection{Distillation of multipartite entanglement for a holographic state}

Entanglement distillation is the \textit{asymptotic} quantum information process that extracts, using local operations and classical communication(LOCC), some number of pure EPR pairs from a large number of copies of a quantum state. For a holographic state, \cite{Bao2019BeyondDistillation,Mori2022distillation} achieved \textit{one-shot} entanglement distillation by constructing a tensor network satisfying some conditions. The geometric counterpart of the holographic entanglement distillation corresponds to pushing a boundary subregion until it matches its RT surface, and distilled EPR pairs are observed to live across the RT surface. Relation to EWCS has been explored in \cite{Gun2024distillation,Mori:2024gwe}. It is interesting to study how to distill, for instance, tripartite and bipartite entanglement captured by the tripartite EWCS $E^{(3)}_w(A:B:C)$.

\subsection{Holographic codes for mixed states and optimally purified states}

Holographic quantum error correction codes have been explored for the purpose of understanding bulk reconstructions in a holographic system in the semiclassical limit\cite{Almheiri2015bulklocality,Harlow2017RTfomulaQEC,Pastawski2015Happy,Pastawski:2015qua,Faulkner:2020hzi,Furuya2022Petz,Chandra2023RTN-HoloCodes,Verlinde2013BH-QEC,Kar2023Non-isometric, Leutheusser2025subsubduality}. 
Recently, its relation to the holographic entropy inequalities has been discussed in \cite{Czech:2025jnw}.

For our future work\cite{BFM:preparation}, we are interested in the role of multipartite entanglement in the properties of the following two types of holographic codes: i) a \textit{mixed code}, a holographic code associated with a mixed state, and ii) a \textit{purified code}, a holographic code associated with an optimally purified state, see figure \ref{fig:QECC}.
\begin{figure}
    \centering
    \includegraphics[width=0.9\linewidth]{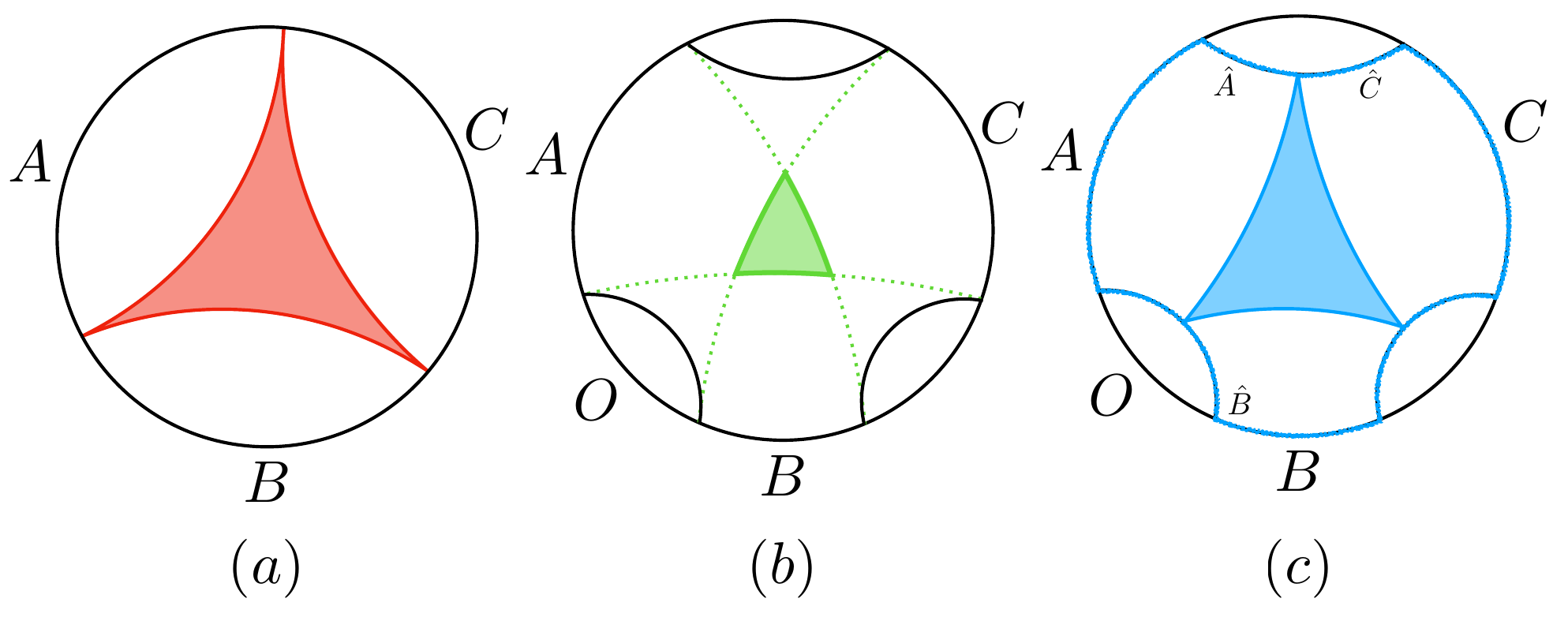}
    \caption{\small{(a) Intersections of bipartite entanglement wedge $AB$, $BC$, $CA$ for a pure holographic state of $ABC$, i.e., $M_{AB}\cap M_{BC}\cap M_{CA}$. (b) Intersections of bipartite entanglement wedge $AB$, $BC$, $CA$ for a mixed holographic state of $ABC$, i.e., $M_{AB}\cap M_{BC}\cap M_{CA}$. (c) A geometry corresponding to an optimally purified holographic state of $ABC$. The edges of the hyperbolic triangle computes $E_w^{(3)}(A:B:C)$.}}
    \label{fig:QECC}
\end{figure}

A holographic quantum error correction code on Figure \ref{fig:QECC} (a) was proposed and studied in \cite{Almheiri2015bulklocality, Harlow2017RTfomulaQEC}. The red region is the intersection of the entanglement wedge $AB$, $BC$, and $AC$. The bulk information within the intersection is non-locally encoded to the boundary. As a result, they are protected under partial tracing a single boundary subregion. This mechanism was demonstrated with a qutrit code.

A similar story can be observed in the case of a mixed state. However, one needs extra care with the different phases of entanglement wedge. We focus on two interesting phases: i) tripartite entanglement wedge $M_{ABC}$ and any bipartite entanglement wedges $M_{AB}$, $M_{BC}$, and $M_{AC}$ are connected, ii) $M_{ABC}$ is connected, but all the bipartite entanglement wedges are disconnected, see figure \ref{fig:ewcs-tri-polygamy} for instance.

For i), as in figure \ref{fig:QECC} (b), the green region is the intersection of the entanglement wedge $AB$, $BC$, and $AC$. The information within the intersection is again non-locally encoded to the boundary subregion $ABC$ when all the bipartite entanglement wedges $M_{AB},M_{BC},M_{CA}$ are connected. 

For ii), consider, for example, the symmetric case studied in section \ref{sec:holography} where the size of the boundary subregions $A$, $B$, and $C$ is equal, and the subregions are placed evenly on the boundary. When the holographic state is in the phase $\alpha^{(2)}\leq \alpha \leq \alpha^{(3)}$, figure \ref{fig:ewcs-tri-polygamy}, any bipartite entanglement wedges no longer exist, although the tripartite entanglement wedge $M_{ABC}$ does. As a result, the holographic code associated with the mixed state of $ABC$ cannot correct against the erasure error of a single boundary region after all the bipartite entanglement wedges become disconnected. In other words, one strictly needs access to all three boundary subregions for the bulk reconstruction of $M_{ABC}$. However, the holographic code associated with the optimally purified state $\ket{\Psi}_{ABCA\hat{A}\hat{B}\hat{C}}$ in figure \ref{fig:QECC} (c) can be correctable under partial tracing $\hat{A}\supset A$, $\hat{B}\supset B$, or $\hat{C}\supset C$. We will explore the properties of a mixed code and purified code in \cite{BFM:preparation}.

A bulk reconstruction for any bulk subregions that are not in contact with the boundary has been found to be possible in the language of von Neumann algebras in \cite{Leutheusser2025subsubduality}. The key relation to the holographic quantum error correction was observed to be the superadditivity, or Haag's duality of bulk and boundary local algebras\cite{Faulkner:2020hzi}. It is a great opportunity to explore how to algebraically extract multipartite entanglement from, for instance, inclusions or intersections of local algebras in a holographic quantum system.

\subsection{More examples}

EoP has been studied in several quantum systems, such as spin systems, free fields, and CFTs\cite{Caputa2019HEoPinCFT,Bhattacharyya2019freefields-spins-EoP,Bhattacharyya2018freescalarEoP,GUO2019134934}. Moreover, extensive studies have been conducted on entanglement of purification and reflected entropy using random tensor networks\cite{AkersEoPinRTNs2024,Akers2022RE-RTNs,Akers2023RE-RTNsII,Akers2024RE-RTNsIII}. It is interesting to calculate $\Delta^{(3)}(A:B:C)$ in these models.   

In addition to our computations for pure AdS$_3$ in section \ref{sec:holography}, it would be interesting to explore the quantity for the case of a BTZ black hole. We conclude the section by briefly commenting on the one-sided BTZ black hole in $2+1$ dimensions.

Its metric is given by
\begin{equation}\label{eq:BTZ-metric}
    ds^2= - \frac{r^2-r_+^2}{L_{AdS}^2}dt^2 +\frac{L_{AdS}^2}{r^2-r_+^2}dr^2+r^2d\phi^2,
\end{equation}
where the temperature $T_H$ is given by $\beta/L_{AdS}=2\pi L_{AdS}/r_+$, and $\beta=1/T_{H}$. There are several possible phases of entanglement wedges. The connected entanglement wedge with non-zero $\Delta^{(3)}(A:B:C)$ is depicted in figure \ref{fig:delta3-ewcs-btz}.

\begin{figure}[h!]
    \centering
    \includegraphics[width=0.4\linewidth]{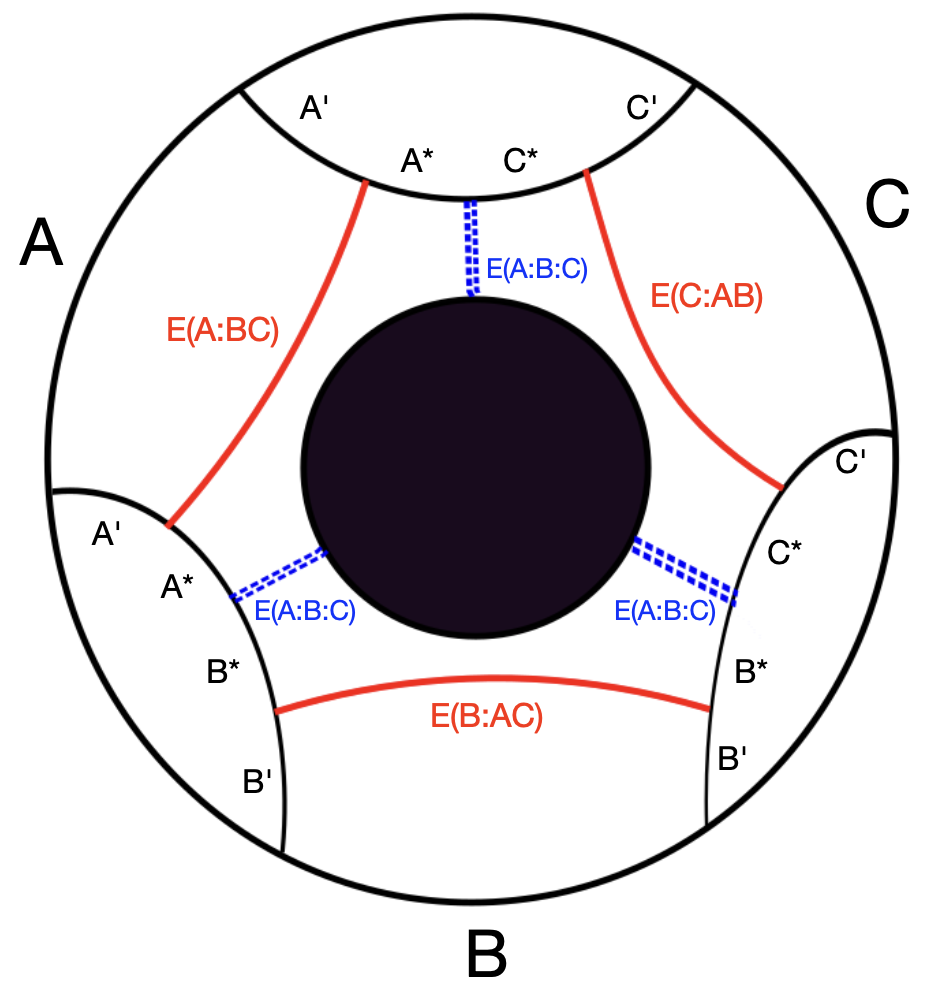}
    \caption{The connected entanglement wedge of $ABC$ in a static time-slice of $AdS_3/CFT_2$. For simplicity, we are considering the symmetric case, where the subregions $A$, $B$ and $C$ have equal angular half-width $\alpha$ (see \ref{eq:regions_ABC}). The red curves denote $E_w^{(2)}(A:BC)$ and permutations, whereas the blue curve denotes $E_w^{(3)}(A:B:C)$. The double lines imply that they are counted twice.}
    \label{fig:delta3-ewcs-btz}
\end{figure}
First, let us make some comments regarding the phases where $\Delta^{(3)}(A:B:C)$ vanishes. We can describe them as follows,
\begin{itemize}
    \item The regions $A$, $B$ and $C$ are small enough such that the entanglement wedge of $ABC$ is disconnected. Analogous to pure $AdS_3$, this corresponds to a separable state.
    \item The blackhole temperature $T_H$ is very large, and thus the horizon radius $r_+$ is also large. In this regime, the optimal purification of $E_{w}^{(2)}$ is the same as that of $E_{w}^{(3)}$, i.e., the regions $A'$ and $A^*$ (and similarly others) coincide. Hence, $\Delta^{(3)}(A:B:C)$ vanishes. In this state, the bipartite correlations between the blackhole and its exterior dominate entirely. 
\end{itemize}

Now we will consider the possible geometries where $\Delta^{(3)}(A:B:C)$ is non-vanishing.
\begin{itemize}
    \item We have not investigated if there is a possibility of an analogous phase similar to the connected phase of pure $AdS_3$, if the blackhole temperature is small enough, just above the threshold. Such a phase (if exists) has non-vanishing $\Delta^{(3)}(A:B:C)$. 
    \item Next we move on to the case, where the temperature is large enough (but not too large), where we have the phase depicted in figure \ref{fig:delta3-ewcs-btz}. We will discuss this phase in more details below.
\end{itemize}

In the figure \ref{fig:delta3-ewcs-btz}, $E_w^{(2)}(A:BC)$ is given by \cite{Nguyen:2017yqw},
\begin{equation}\label{eq:btz-Ep2}
\footnotesize
\begin{aligned}
     & E^{(2)}_w(A:BC) \\ & = \frac{2L_{\text{AdS}}}{4G_N} \log \left\{ \frac{\left[ \sqrt{\sinh\left(\frac{r_+}{2L_{\text{AdS}}}(\phi_1 - \phi_3)\right) \sinh\left(\frac{r_+}{2L_{\text{AdS}}}(\phi_2 - \phi_4)\right)} + \sqrt{\sinh\left(\frac{r_+}{2L_{\text{AdS}}}(\phi_2 - \phi_1)\right) \sinh\left(\frac{r_+}{2L_{\text{AdS}}}(\phi_4 - \phi_3)\right)} \right]^2 }
    {\sinh\left(\frac{r_+}{2L_{\text{AdS}}}(\phi_2 - \phi_3)\right) \sinh\left(\frac{r_+}{2L_{\text{AdS}}}(\phi_1 - \phi_4)\right)} \right\}
\end{aligned}
\end{equation}

where the angles can be parameterized by the half-width $\alpha$ as
\begin{equation}\label{eq:phi-alpha}
    (\phi_1,\phi_2,\phi_3,\phi_4)=\left(2\pi - \alpha, \alpha, \pi - \left(\frac{\pi}{3} + \alpha\right), \pi + \left(\frac{\pi}{3} + \alpha\right)\right).
\end{equation}

To find $E(A:B:C)$, we have to find the length of the dashed blue curve from the turning point of the RT surface (e.g., the point separating $A^*$ and $C^*$) to the blackhole horizon. Let this point have a radial distance $r_*$, then $E(A:B:C)$ is given by \cite{Nguyen:2017yqw}
\begin{equation}
    E^{(3)}_w(A:B:C)=\frac{6L_{AdS}}{4G_N}\log{\left(\frac{r_*}{r_+}+\sqrt{\left(\frac{r_*}{r_+}\right)^2-1}\right)},
\end{equation}
where one can parameterize $r_*$ with $\alpha$. This phase can exist as long as we have the condition $E(A:BC)<\frac{1}{3}E(A:B:C)$, which after switching gives us $\Delta^{(3)}(A:B:C)=0$.

\section*{Acknowledgement}
N.B. thanks Brianna Grado-White, Donald Marolf, and Wayne W. Weng for initial collaboration. We also thank Jacob March and Yikun Jiang for fruitful discussions. N. B. is supported by the DOE Office of Science-ASCR, in particular the grant Novel Quantum Algorithms from Fast Classical Transforms. K.F. is supported by N.B.'s startup fund at Northeastern University. J.N. is supported by the graduate teaching assistantship at Northeastern University.

\appendix

\section{Properies of multipartite entanglement of purification(MEoP)} \label{app:MEoP}

Here, we list the properties of MEoP from \cite{Umemoto-2018-MultipartiteEoP,Bao-2019-conditionalEoP}.

\begin{lemma}[Properties of MoP\cite{Umemoto-2018-MultipartiteEoP,Bao-2019-conditionalEoP}]\
    \begin{enumerate}
    \item If $\rho_{A_1\cdots A_n} = \rho_{A_1 \cdots A_{n-1}}\otimes \rho_{A_n}$, then
    \begin{equation}
        E^{(n)}_p(A_1:\cdots:A_n) = E^{(n-1)}_p(A_1:\cdots:A_{n-1})
    \end{equation}
    \item If $\rho_{A_1\cdots A_n}$ is pure,
    \begin{equation}\label{eq:EoP-pure}
        E^{(n)}_p(A_1:\cdots:A_n) = \sum_{i=1}^nS_{A_i}.
    \end{equation}
    \item 
    \begin{equation}\label{eq:EoP-product}
        E^{(n)}_p(A_1:\cdots:A_n) = 0 \text{ iff } \rho_{A_1\cdots A_n} = \rho_{A_1} \otimes \cdots \otimes \rho_{A_n}
    \end{equation}
    \item Monotonicity under a local partial trace
    \begin{equation}\label{eq:MEoP_monotonicity}
        E^{(n)}_p(A_1:\cdots:A_nB) \geq E^{(n)}_p(A_1:\cdots:A_n)  
    \end{equation}
    \item For any $n$-partite state $\rho_{A_1\cdots A_n}$,
    \begin{equation}\label{eq:n_n-1}
        E^{(n)}_p(A_1:\cdots : A_{n-1} :A_n) \geq E^{(n-1)}_p(A_1:\cdots : A_{n-1}A_n)
    \end{equation}
    \item For any pure $n$-partite density matrices $\rho_{A_1\cdots A_n}$, $E^{(n)}_p$ is polygamous, i.e.,
    \begin{equation}\label{eq:polygamy-MEoP}
        E^{(n)}_p(A_1:\cdots:BC) \leq E^{(n)}_p(A_1:\cdots:B) +E^{(n)}_p(A_1:\cdots:C)
    \end{equation}
    \item The upper bound of MEoP is given by
    \begin{equation}
        E^{(n)}_p \leq min_i \{S_{A_1} + \cdots+ S_{A_1\cdots A_{i-1}A_{i+1}}\cdots A_n + \cdots+ S_{A_n}\}
    \end{equation}
    For instance,
    \begin{equation}
        E^{(3)}_p(A:B:C) \leq min\{A+B+AB,B+C+BC,C+A+CA\}
    \end{equation}

    \item One lower bound of MEoP is given by
    \begin{equation}\label{eq:lowerbound_Epn}
        E^{(n)}_p(A_1:..:A_n) \geq  \frac{1}{2}\sum_{i=1}^n E_p^{(2)}(A_i: A_1 \cdots A_{i-1}A_{i+1}\cdots A_{n})
    \end{equation}

    \item Another lower bound of MEoP is given by
    \begin{equation}\label{eq:lowerbound_2_Epn}
        E^{(n)}_p(A_1:..:A_n) \geq  I(A_1:\cdots:A_n)
    \end{equation}
    where $I(A_1:\cdots:A_n)$ is a multipartite mutual information defined as
    \begin{equation}
        I(A_1:\cdots:A_n) := S(\rho_{A_1\cdots A_n}\| \rho_{A_1} \otimes \cdots \otimes \rho_{A_n}) = \sum_i S_{A_i} - S_{A_1\cdots A_n}
    \end{equation}
    
    \item $E^{(3)}_p(A:B:C)$ is bounded from below by
    \begin{equation}
        E^{(3)}_p(A:B:C) \geq  max\{S_A+S_B+S_C-S_{ABC}, 2(S_A+S_B+S_C)-S_{AB}-S_{BC}-S_{CA}\}.
    \end{equation}

\end{enumerate}
\end{lemma}

\section{Four partite correlation signal $\Delta_p^{(4)}$}\label{app:delta4p}

In this section, we prove the properties of $\Delta^{(4)}_p(A:B:C:D)$, listed in proposition \ref{prop:delta4p}, in the subsequent subsections. In the last two subsections, we study the properties of $\Delta^{(4)}_p(A:B:C:D)$ for GHZ$_N$ states for $N\geq 4$ and four-partite holographic pure state.


\subsection{Coefficients}
In general, $\Delta^{(4)}(A:B:C:D)$ can be written with coefficients as follows.
\begin{equation}\label{eq:candidate_delta4p}
\begin{split}
    \Delta^{(4)}&(A:B:C:D) \\
    =& E^{(4)}_p(A:B:C:D) - e^{(3)} \mbE^{(3)}_p(2:1:1) - e^{(2)}_1\mbE^{(2)}_p(1:3) - e^{(2)}_2\mbE^{(2)}_p(2:2) \\
    =& E^{(4)}_p(A:B:C:D) - d^{(3)} \boldsymbol{\Delta}^{(3)}_p(2:1:1) - d^{(2)}_1 \boldsymbol{\Delta}^{(2)}_p(1:3) - d^{(2)}_2\boldsymbol{\Delta}^{(2)}_p(2:2).\\
\end{split}
\end{equation}
Here, we introduced the simplified notations, i.e.,
\begin{equation}\label{eq:notation_Ep_combinations}
\begin{split}
    \mbE^{(3)}_p(2:1:1) :=& E^{(3)}_p(AB:C:D) +E^{(3)}_p(AC:B:D)+E^{(3)}_p(AD:B:C) \\
    &+E^{(3)}_p(BC:A:D)+E^{(3)}_p(BD:A:C)+E^{(3)}_p(CD:A:B),\\
    \mbE^{(2)}_p(1:3) := & E^{(2)}_p(A:BCD) + E^{(2)}_p(B:ACD) + E^{(2)}_p(C:ABD) + E^{(2)}_p(D:ABC),\\
    \mbE^{(2)}_p(2:2) := &E^{(2)}_p(AB:CD) + E^{(2)}_p(AC:BD) +E^{(2)}_p(AD:BC),\\
\end{split}
\end{equation}
and 
\begin{equation}
\begin{split}
    \boldsymbol{\Delta}^{(3)}_p(2:1:1) :=& \Delta^{(3)}_p(AB:C:D) +\Delta^{(3)}_p(AC:B:D)+\Delta^{(3)}_p(AD:B:C) \\
    & + \Delta^{(3)}_p(BC:A:D)+\Delta^{(3)}_p(BD:A:C)+\Delta^{(3)}_p(CD:A:B),\\
    \boldsymbol{\Delta}^{(2)}_p(1:3) := & \Delta^{(2)}_p(A:BCD) + \Delta^{(2)}_p(B:ACD) + \Delta^{(2)}_p(C:ABD) + \Delta^{(2)}_p(D:ABC),\\
    \boldsymbol{\Delta}^{(2)}_p(2:2) := &\Delta^{(2)}_p(AB:CD) + \Delta^{(2)}_p(AC:BD) +\Delta^{(2)}_p(AD:BC).\\
\end{split}
\end{equation}
The coefficients are related by
\begin{equation}\label{eq:coefficients-e-d}
    e^{(3)} = d^{(3)},\; e^{(2)}_1 = d^{(2)}_1 - \frac{3}{2}d^{(3)},\; e^{(2)}_2 = d^{(2)}_2- d^{(3)}
\end{equation}

We determine the coefficients $\{e^{(3)}, e^{(2)}_1,e^{(2)}_2\}$ or $\{d^{(3)}, d^{(2)}_1,d^{(2)}_2\}$ by requiring
\begin{equation}\label{eq:coefficient_condition_delta_4}
    \Delta^{(4)}_p(A:B:C:D) = 0 
\end{equation}
for the quantum states that only have $k$-partite entanglement for $k=2,3$, such as $\rho_{ABCD} = \rho_{AB}\otimes \rho_C \otimes \rho_D$, $\rho_{ABCD} = \rho_{ABC}\otimes \rho_D$, $\rho_{ABCD} = \rho_{AB}\otimes \rho_{CD}$, and others with any permutations among $A,B,C,D$. We write down $\Delta^{(4)}_p(A:B:C:D)$ with the coefficients as a definition below.

\begin{definition}[A candidate four-partite correlation signal for a density matrix]\label{def:delta_p_4}

    For a mixed density matrix $\rho_{ABCD}$, $\Delta^{(4)}_p(A:B:C:D)$ is defined as
    \begin{equation}\label{eq:delta_p_4}
    \begin{split}
        \Delta^{(4)}_p&(A:B:C:D) \\
        :=& E^{(4)}_p(A:B:C:D)\\
        &\quad - \frac{1}{3} \{E^{(3)}_p(AB:C:D) +E^{(3)}_p(AC:B:D)+E^{(3)}_p(AD:B:C) \\
        &\qquad \qquad+E^{(3)}_p(BC:A:D)+E^{(3)}_p(BD:A:C)+E^{(3)}_p(CD:A:B)\} \\
        &\quad + \frac{1}{6}\{E^{(2)}_p(A:BCD) + E^{(2)}_p(B:ACD) + E^{(2)}_p(C:ABD) + E^{(2)}_p(D:ABC)\}\\
        &\quad + \frac{1}{6}\{E^{(2)}_p(AB:CD) + E^{(2)}_p(AC:BD) +E^{(2)}_p(AD:BC)\} \\
        :=& E^{(4)}_p(A:B:C:D) \\
        &\quad - \frac{1}{3} \{\Delta^{(3)}_p(AB:C:D) +\Delta^{(3)}_p(AC:B:D)+\Delta^{(3)}_p(AD:B:C) \\
        &\qquad \qquad+\Delta^{(3)}_p(BC:A:D)+\Delta^{(3)}_p(BD:A:C)+\Delta^{(3)}_p(CD:A:B)\} \\
        &\quad - \frac{1}{3}\{\Delta^{(2)}_p(A:BCD) + \Delta^{(2)}_p(B:ACD) + \Delta^{(2)}_p(C:ABD) + \Delta^{(2)}_p(D:ABC)\}\\
        &\quad - \frac{1}{6}\{\Delta^{(2)}_p(AB:CD) + \Delta^{(2)}_p(AC:BD) +\Delta^{(2)}_p(AD:BC)\}.\\
    \end{split}
    \end{equation}
    
\end{definition}

\begin{proof}[Derivation]
For $\rho_{ABCD} = \rho_{AB}\otimes \rho_C \otimes \rho_D$, the relevant entropies are
\begin{align}
\begin{split}
    S_{AB\tilde{A}\tilde{B}}&=S_{CD\tilde{C}\tilde{D}}=S_{C\tilde{C}} = S_{D\tilde{D}} = 0,\\
    S_{AC\tilde{A}\tilde{C}} &=S_{AD\tilde{A}\tilde{D}} = S_{BC\tilde{B}\tilde{C}} =S_{BD\tilde{B}\tilde{D}}=S_{A\tilde{A}}= S_{B\tilde{B}}\\
\end{split}
\end{align}
Then, we have\footnote{We denote $\underset{\ket{\psi}_{ABCD\tilde{A}\tilde{B}\tilde{C}\tilde{D}}}{min}$ as $min$ for brevity in this derivation.}
\begin{align}
\begin{split}
    E^{(4)}_p (A:B:C:D) &= min [S_{A\tilde{A}} + S_{B\tilde{B}}], \\
    \mbE^{(3)}_p(2:1:1)  &= 5 min [S_{A\tilde{A}} + S_{B\tilde{B}}],\\
    \mbE^{(2)}_p(1:3) &= 2min[S_{A\tilde{A}} + S_{B\tilde{B}}],\\
    \mbE^{(2)}_p(2:2)  &= 2min[S_{A\tilde{A}} + S_{B\tilde{B}}].\\
\end{split}
\end{align}
From \eqref{eq:candidate_delta4p},
\begin{equation}\label{eq:delta4_bipartite}
    \Delta^{(4)}_p(A:B:C:D) = (1-5e^{(3)} - 2e^{(2)}_1 - 2e^{(2)}_2)min[S_{A\tilde{A}} + S_{B\tilde{B}}]. 
\end{equation}
Similarly for the case of $\rho_{ABC}\otimes \rho_D$, we obtain
\begin{align}
\begin{split}
S_{ABC\tilde{A}\tilde{B}\tilde{C}}&=S_{D\tilde{D}} = 0,\;S_{AD\tilde{A}\tilde{D}} = S_{A\tilde{A}},\;S_{BD\tilde{B}\tilde{D}} = S_{B\tilde{B}}, \;S_{CD\tilde{C}\tilde{D}} = S_{C\tilde{C}},\\
\end{split}
\end{align}
and 
\begin{align}
\begin{split}
    E^{(4)}_p (A:B:C:D) &= min [S_{A\tilde{A}} + S_{B\tilde{B}}+S_{C\tilde{C}}], \\
    \mbE^{(3)}_p(2:1:1)  &= 3 min [S_{A\tilde{A}} + S_{B\tilde{B}}+ S_{C\tilde{C}}] + min[2S_{A\tilde{A}}] +min[2S_{B\tilde{B}}]+min[2S_{C\tilde{C}}] ,\\
    \mbE^{(2)}_p(1:3) &= \mbE^{(2)}_p(2:2)  = min[2S_{A\tilde{A}}] +min[2S_{B\tilde{B}}]+min[2S_{C\tilde{C}}] .\\
\end{split}
\end{align}
Then,
\begin{align}\label{eq:delta4_tripartite}
\begin{split}
    \Delta^{(4)}_p&(A:B:C:D) \\
    &= (1-3e^{(3)}) min[S_{A\tilde{A}} + S_{B\tilde{B}}+S_{C\tilde{C}}] \\
    &+(-e^{(3)}-e^{(2)}_1-e^{(2)}_2) \{min[2S_{A\tilde{A}}] +min[2S_{B\tilde{B}}]+min[2S_{C\tilde{C}}]\}.
\end{split}
\end{align}
For the third case $\rho_{ABC}= \rho_{AB}\otimes \rho_{CD}$, we have
\begin{align}
\begin{split}
S_{AB\tilde{A}\tilde{B}}&=S_{CD\tilde{C}\tilde{D}} = 0,\\
S_{AC\tilde{A}\tilde{C}}&=S_{BC\tilde{B}\tilde{C}}= S_{AD\tilde{A}\tilde{D}} =S_{BD\tilde{B}\tilde{D}} =S_{A\tilde{A}}+ S_{C\tilde{C}} =S_{B\tilde{B}}+S_{D\tilde{D}},\\
\end{split}
\end{align}
where we used the purification symmetries $S_{A\tilde{A}}=S_{B\tilde{B}}$ and $S_{C\tilde{C}} =S_{D\tilde{D}}$. The relevant MEoPs can be computed as 
\begin{align}
\begin{split}
    E^{(4)}_p (A:B:C:D) &= min [S_{A\tilde{A}}+S_{B\tilde{B}} +  S_{C\tilde{C}}+S_{D\tilde{D}}], \\
    \mbE^{(3)}_p(2:1:1)  &= 4min [S_{A\tilde{A}}+S_{B\tilde{B}} +  S_{C\tilde{C}}+S_{D\tilde{D}}] \\
    &\quad + min [S_{A\tilde{A}}+S_{B\tilde{B}}]+ min [S_{C\tilde{C}}+S_{D\tilde{D}}],\\
    \mbE^{(2)}_p(1:3) &= 2min [S_{A\tilde{A}}+S_{B\tilde{B}}]+ 2min [S_{C\tilde{C}}+S_{D\tilde{D}}]\\
    \mbE^{(2)}_p(2:2)  &= 2min [S_{A\tilde{A}}+S_{B\tilde{B}} +  S_{C\tilde{C}}+S_{D\tilde{D}}]  .\\
\end{split}
\end{align}
Then, 
\begin{align}\label{eq:delta4_bipartite_bipartite}
\begin{split}
    \Delta^{(4)}_p&(A:B:C:D) \\
    &= (1-4e^{(3)} - 2e^{(2)}_2) min[S_{A\tilde{A}}+S_{B\tilde{B}} +  S_{C\tilde{C}}+S_{D\tilde{D}}] \\
    &+(-e^{(3)}-2e^{(2)}_1) \{min[S_{A\tilde{A}}+S_{B\tilde{B}}] +min[S_{C\tilde{C}}+S_{D\tilde{D}}]\}.
\end{split}
\end{align}
In summary, by imposing the condition \eqref{eq:coefficient_condition_delta_4} to \eqref{eq:delta4_bipartite}, \eqref{eq:delta4_tripartite}, and \eqref{eq:delta4_bipartite_bipartite}, we have, 
\begin{itemize}
    \item for $\rho_{ABCD} = \rho_{AB}\otimes \rho_C \otimes \rho_D$,
    \begin{equation}\label{eq:condition_delta4_1-1}
        (1-5e^{(3)} - 2e^{(2)}_1 - 2e^{(2)}_2)min[S_{A\tilde{A}} + S_{B\tilde{B}}] = 0,
    \end{equation}
    
    \item for $\rho_{ABCD} = \rho_{ABC}\otimes \rho_D$,
    \begin{equation}\label{eq:condition_delta4_2-1}
    \begin{split}
        &(1-3e^{(3)}) min[S_{A\tilde{A}} + S_{B\tilde{B}}+S_{C\tilde{C}}] \\
        &+(-e^{(3)}-e^{(2)}_1-e^{(2)}_2) \{min[2S_{A\tilde{A}}] +min[2S_{B\tilde{B}}]+min[2S_{C\tilde{C}}]\}=0,
    \end{split}
    \end{equation}
    
    \item and for $\rho_{ABCD} = \rho_{AB}\otimes \rho_{CD}$,
    \begin{equation}\label{eq:condition_delta4_3-1}
    \begin{split}
        &(1-4e^{(3)} - 2e^{(2)}_2) min[S_{A\tilde{A}}+S_{B\tilde{B}} +  S_{C\tilde{C}}+S_{D\tilde{D}}] \\
        &+(-e^{(3)}-2e^{(2)}_1) \{min[S_{A\tilde{A}}+S_{B\tilde{B}}] +min[S_{C\tilde{C}}+S_{D\tilde{D}}]\} = 0.
    \end{split}
    \end{equation}

\end{itemize}

The first equation \eqref{eq:condition_delta4_1-1} becomes 
\begin{equation}\label{eq:condition_delta4_1-2}
    (1-5e^{(3)} - 2e^{(2)}_1 - 2e^{(2)}_2) = 0
\end{equation}
because, in general, $min[S_{A\tilde{A}} + S_{B\tilde{B}}]\geq 0$ by the assumption $\rho_{ABCD} = \rho_{AB}\otimes \rho_C\otimes \rho_D$. The second equation \eqref{eq:condition_delta4_2-1} can be reorganized into 
\begin{equation}\label{eq:condition_delta4_2-2}
\begin{split}
    &(1-3e^{(3)}) \Delta^{(3)}_p(A:B:C) \\
    &+(\frac{1}{2}-\frac{5}{2}e^{(3)}-e^{(2)}_1-e^{(2)}_2) \{min[2S_{A\tilde{A}}] +min[2S_{B\tilde{B}}]+min[2S_{C\tilde{C}}]\}=0.
\end{split}
\end{equation}
Here, we applied that
\begin{equation}
    \Delta^{(3)}_p(A:B:CD) =\Delta^{(3)}_p(A:B:C)
\end{equation}
for $\rho_{ABCD}=\rho_{ABC}\otimes \rho_D$. From \eqref{eq:condition_delta4_1-2}, the equation \eqref{eq:condition_delta4_2-2} becomes 
\begin{equation}\label{eq:condition_delta4_2-3} 
    (1-3e^{(3)}) \Delta^{(3)}_p(A:B:C) = 0.
\end{equation}
Because $\Delta^{(3)}_p(A:B:C) \neq 0 $ by the corresponding assumption $\rho_{ABC} =\rho_{ABC}\otimes \rho_D$, we get
\begin{equation}
    e^{(3)} = \frac{1}{3}.
\end{equation}
Together with the above and \eqref{eq:condition_delta4_1-2} implies that
\begin{equation}
    e^{(2)}_1+e^{(2)}_2 = -\frac{1}{3}.
\end{equation}
Inserting the above equation for $e^{(2)}_1$ and $e^{(2)}_2$ into the third equation \eqref{eq:condition_delta4_3-1} provides
\begin{equation}
\begin{split}
    &(\frac{1}{3}+2e^{(2)}_1) \\
    &\times \{ min[S_{A\tilde{A}}+S_{B\tilde{B}} +  S_{C\tilde{C}}+S_{D\tilde{D}}] -
        min[S_{A\tilde{A}}+S_{B\tilde{B}}] -min[S_{C\tilde{C}}+S_{D\tilde{D}}]\} = 0.
\end{split}
\end{equation}
In general, we have 
\begin{equation}
    min[S_{A\tilde{A}}+S_{B\tilde{B}} +  S_{C\tilde{C}}+S_{D\tilde{D}}] \geq
        min[S_{A\tilde{A}}+S_{B\tilde{B}}] +min[S_{C\tilde{C}}+S_{D\tilde{D}}].
\end{equation}
Hence, without loss of generality, we get
\begin{equation}
    e^{(2)}_1 = -\frac{1}{6}.
\end{equation}
Therefore, 
\begin{equation}\label{eq:coefficients_es}
    e^{(3)}=\frac{1}{3}, \; e^{(2)}_1=-\frac{1}{6},\;e^{(2)}_2=-\frac{1}{6},
\end{equation}
or
\begin{equation}\label{eq:coefficients_ds}
    d^{(3)}=\frac{1}{3}, \; d^{(2)}_1=\frac{1}{3},\;e^{(2)}_2=\frac{1}{6}.
\end{equation}
\end{proof}


\subsection{Purity}

As discussed above, $\Delta_p^{(4)}(A:B:C:D) = 0$ for any density matrices without four-partite entanglement. If $\rho_{ABCD}$ is a pure state, by applying \eqref{eq:EoP-pure} to $\Delta_p^{(4)}(A:B:C:D)$, we find that it does not necessarily vanish, i.e.,
\begin{equation}
    \Delta_p^{(4)}(A:B:C:D) = \frac{1}{3} \{S_A+S_B+S_C+S_D- (S_{AB}+S_{AC}+S_{AD})\} \neq 0 .
\end{equation}

\subsection{Example: $GHZ_N$}

Consider $\ket{GHZ_N}$ for $N\geq 4$ in \eqref{eq:GHZn}, $E^{(4)}_p(A:B:C:D)$ cannot be specified by its upper and lower bound, in contrast to $E^{(2)}_p(A:BC)$ and $E^{(3)}_p(A:B:C)$ in \eqref{eq:GHZ_N_EPs}, since they do not match, i.e.,
\begin{equation}
    7\ln 2 \geq E^{(4)}_p(A:B:C:D) \geq 3 \ln 2.
\end{equation}
With a simple calculation, we have 
\begin{equation}
    \mbE^{(3)}_p(2:1:1) = 6 \times 3\ln 2,\; \mbE^{(2)}_p(1:3) = 4 \times 2\ln 2,\;\mbE^{(2)}_p(2:2) = 3 \times 2\ln 2.
\end{equation}
Then, we find
\begin{equation}
    \Delta^{(4)}_p(A:B:C:D) = E^{(4)}_p(A:B:C:D) - \frac{11}{6} \ln 2.
\end{equation}
Hence, it is simple to see that $\Delta^{(4)}_p(A:B:C:D)\geq 0$ for $\ket{GHZ_N}$ for $N\geq 4$. 

Moreover, if $\rho_{ABCD}=\ket{GHZ_N}\bra{GHZ_N}$ is pure and   $|A|+|B|+|C|+|D|=N\geq 4$, we can compute $\Delta^{(4)}_p(A:B:C:D)$ from \eqref{eq:delta_p4_pure_state}. Since the entanglement entropies of reduced states of any subsystem size of $\ket{GHZ_N}$ is constant, 
\begin{equation}
    S_A =S_B=S_C =S_D=S_{AB} =S_{AC}=S_{AD} =S_{BC}=S_{BD}=S_{CD} = \ln 2,
\end{equation}
we have
\begin{equation}
    \Delta^{(4)}_p(A:B:C:D) = \frac{1}{3}\ln2.
\end{equation}
In contrast to $\Delta^{(3)}(A:B:C)=0$ for a pure state $\rho_{ABC}=\ket{GHZ_N}\bra{GHZ_N}$ and $|A|+|B|+|C|=N$ for any $N$, $\Delta^{(4)}_p(A:B:C:D)$ can capture $GHZ_N$ entanglement for $N\geq 4$. \newline

\subsection{Example: Four-partite holographic pure states}

We saw in proposition \ref{prop:delta4p} that, for a four-partite holographic pure state, $\Delta^{(4)}_p(A:B:C:D)$ is proportional to tripartite information $I_3(X:Y:Z)$. Tripartite information $I_3(X:Y:Z)$ is negative for any holographic states because they satisfy the monogamy of mutual information, e.g.,
\begin{equation}
    S_{XY} +S_{XZ} + S_{YZ} \geq S_{X} + S_{Y}+S_{Z} +S_{XYZ}.
\end{equation}

Hence, we have 
\begin{equation}
    \Delta^{(4)}_p(A:B:C:D) = \frac{1}{3} I_3(X:Y:Z) \leq 0 
\end{equation}
for $X,Y,Z = \{A,B,C,D\}$.

\section{Comments on classical correlations}

In this paper, we did not discuss when a quantum state has classical correlations, e.g., $\rho_{ABC} = \sum_i p_i \rho_{ABC,i}$. As we saw in proposition \ref{prop:delta_p_3_properties}, we confirmed that $\Delta^{(3)}_p(A:B:C) = 0$ for $\rho_{AB}\otimes \rho_C$. In this section, we first show that $\Delta^{(3)}_p(A:B:C) \neq 0$ when classical correlations are present, i.e.,
\begin{equation}\label{eq:classical-bipartite-entanglement}
    \rho_{ABC} = \sum_i p_i \rho_{AB,i}\otimes \rho_{C_i}.
\end{equation}
Then, we devise an alternative quantity $\hat{\Delta}^{(3)}_p(A:B:C)$, which vanishes for the quantum state \eqref{eq:classical-bipartite-entanglement}. However, we do not have proof for $\hat{\Delta}^{(3)}_p(A:B:C) \geq 0$ for any quantum state.

First, for the quantum state \eqref{eq:classical-bipartite-entanglement}, we consider some purified state of $\rho_{ABC}$,
\begin{equation}
    \ket{\psi_{ABC}}\bra{\psi_{ABC}} = \sum_i p_i \ket{\psi_{AB\tilde{A}\tilde{B},i}}\bra{\psi_{AB\tilde{A}\tilde{B},i}} \otimes \ket{\phi_{C\tilde{C},i}}\bra{\phi_{C\tilde{C},i}},    
\end{equation}
such that 
\begin{equation}\label{eq:orthogonality}
    \braket{\phi_{C\tilde{C},i}|\phi_{C\bar{C},j}} = \delta_{ij} \text{ and }\braket{\psi_{AB\tilde{A}\tilde{B},i}|\psi_{AB\tilde{A}\tilde{B},j}} = \delta_{ij}.
\end{equation}
    
Here, $\tilde{A}\tilde{B}\tilde{C}$ is the purifier of $ABC$. By the purity, the von Neumann entropies $S_{A\tilde{A}},S_{B\tilde{B}}, S_{BC\tilde{B}\tilde{C}}, S_{AC\tilde{A}\tilde{C}}$ of the reduced states $\sigma_{A\tilde{A}}, \sigma_{B\tilde{B}},\sigma_{BC\tilde{B}\tilde{C}}, \sigma_{AC\tilde{A}\tilde{C}}$ satisfy
\begin{equation}
    S_{A\tilde{A}} = S_{BC\tilde{B}\tilde{C}},\; S_{B\tilde{B}} =S_{AC\tilde{A}\tilde{C}}
\end{equation}
The reduced states, $\sigma_{BC\tilde{B}\tilde{C}}, \sigma_{AC\tilde{A}\tilde{C}}$ can be written as 
\begin{equation}
    \sigma_{BC\tilde{B}\tilde{C}} = \sum_{i} p_i \sigma_{B\tilde{B},i} \otimes \ket{\phi_{C\tilde{C},i}}\bra{\phi_{C\tilde{C},i}},\; \sigma_{AC\tilde{A}\tilde{C}} = \sum_{i} p_i \sigma_{A\tilde{A},i} \otimes \ket{\phi_{C\tilde{C},i}}\bra{\phi_{C\tilde{C},i}}.
\end{equation}
Together with the orthogonality \eqref{eq:orthogonality}, we have \begin{equation}
\begin{split}
    S_{A\tilde{A}}&=S_{BC\tilde{B}\tilde{C}} = \sum_i p_i S_{B\tilde{B},i} + H(\{p\}_\rho),\\ S_{B\tilde{B}} &= S_{AC\tilde{A}\tilde{C}}  =\sum_i p_i S_{A\tilde{A},i} + H(\{p\}_\rho),\\ S_{C\tilde{C}} &= S_{AB\tilde{A}\tilde{B}}  = H(\{p\}_\rho)\\
\end{split}
\end{equation}
where $H(\{p\}_\rho)$ is the Shannon entropy of the probability distribution $\{p\}_\rho:= \{p_i, \forall i\}$.
The bipartite and tripartite EoP can be computed as\footnote{We denote $\underset{\ket{\psi}_{ABC\tilde{A}\tilde{B}\tilde{C}}}{min}$ as $min$ for simplicity in this section.}
\begin{equation}
\begin{split}
    E^{(2)}_p(A:BC) &= 2 min \Big[\sum_i p_i S_{B\tilde{B},i} \Big]+ 2H(\{p\}_\rho),\\
    E^{(2)}_p(B:AC) &= 2 min \Big[\sum_i p_i S_{A\tilde{A},i}\Big]+ 2H(\{p\}_\rho),\\
    E_p^{(2)}(C:AB)&=2H(\{p\}),\\
\end{split}
\end{equation}
and
\begin{equation}
    E_p^{(3)}(A:B:C)= min\Big[\sum_i p_i S_{A\tilde{A},i} +\sum_i p_i S_{B\tilde{B},i}\Big]+ 3H(\{p\}_\rho).
\end{equation}
Here, we used the fact that the Shannon entropies $H(\{p\}_\rho)$ are independent of minimization. Hence, we have
\begin{equation}
    \Delta^{(3)}_p(A:B:C) = min\Big[\sum_i p_i S_{A\tilde{A},i}+ \sum_i p_i S_{B\tilde{B},i} \Big] -  min\Big[\sum_i p_i S_{A\tilde{A},i}\Big] - min\Big[\sum_i p_i S_{B\tilde{B},i} \Big] \geq 0,
\end{equation}
which is positive semidefinite. Thus, $\Delta^{(3)}_p(A:B:C)$ does not vanish if the classical correlations exist in a quantum state.

Now, we briefly study an experimental quantity that vanishes for the quantum state \eqref{eq:classical-bipartite-entanglement}. We define the quantity as
\begin{equation}
\begin{split}
    \hat{\Delta}^{(3)}_p(A:B:C) :=& E^{(3)}_p(A:B:C) + \frac{1}{2}\{E^{(2)}_p(A:BC)+E^{(2)}_p(B:AC)+ E_p^{(2)}(C:AB)\}\\
    &-\{E^{(3)}_p(A:B:\slashed{C}) + E^{(3)}_p(A:\slashed{B}:C)+E^{(3)}_p(\slashed{A}:B:C)\}.\\
\end{split}
\end{equation}
Here, we introduced
\begin{equation}
\begin{split}
    E^{(3)}_p(A:B:\slashed{C}) :=& \underset{\ket{\psi}_{ABC\tilde{A}\tilde{B}\tilde{C}}}{min} [S_{A\tilde{A}} + S_{B\tilde{B}}],\\
    E^{(3)}_p(A:\slashed{B}:C) :=& \underset{\ket{\psi}_{ABC\tilde{A}\tilde{B}\tilde{C}}}{min} [S_{A\tilde{A}} + S_{C\tilde{C}}],\\
    E^{(3)}_p(\slashed{A}:B:C) :=& \underset{\ket{\psi}_{ABC\tilde{A}\tilde{B}\tilde{C}}}{min} [S_{B\tilde{B}} + S_{C\tilde{C}}].\\
\end{split}
\end{equation}
These quantities satisfy, e.g.,
\begin{equation}
    E^{(3)}_p(A:B:C) \geq E^{(3)}_p(A:B:\slashed{C}) \geq E^{(2)}_p(A:B).
\end{equation}
It is simple to see the first inequality. The second inequality can be proved as follows. Suppose $\ket{\psi'}_{ABC\tilde{A}\tilde{B}\tilde{C}}$ is an optimal purification for $\rho_{ABC}$ so that
\begin{equation}
    E^{(3)}_p(A:B:\slashed{C}) = S_{A\tilde{A}} + S_{B\tilde{B}}.
\end{equation}
This state is one of the purifications $\ket{\phi_{AB\tilde{A}\tilde{B}}}=\ket{\psi'}_{ABC\tilde{A}\tilde{B}\tilde{C}}$ for $\rho_{AB}$. Hence, 
\begin{equation}
    E^{(3)}_p(A:B:\slashed{C}) = S_{A\tilde{A}} + S_{B\tilde{B}}\geq \underset{\ket{\phi}_{AB\tilde{A}\tilde{B}}}{min} [S_{A\tilde{A}} + S_{B\tilde{B}}] = E^{(2)}_p(A:B).
\end{equation}

For the quantum state \eqref{eq:classical-bipartite-entanglement}, we have
\begin{equation}
\begin{split}
    E^{(3)}_p(A:B:\slashed{C}) =& min\Big[\sum_i p_i S_{A\tilde{A},i} +\sum_i p_i S_{B\tilde{B},i}\Big]+ 2H(\{p\}_\rho),\\
    E^{(3)}_p(A:\slashed{B}:C) =& min\Big[\sum_i p_i S_{B\tilde{B},i}\Big]+ 2H(\{p\}_\rho),\\
    E^{(3)}_p(\slashed{A}:B:C) =& min\Big[\sum_i p_i S_{A\tilde{A},i} \Big]+ 2H(\{p\}_\rho).\\
\end{split}
\end{equation}
Then, we can see that
\begin{equation}
    \hat{\Delta}^{(3)}_p(A:B:C) = 0.
\end{equation}

It is simple to see that it vanishes for a mixed product state as well, for instance, $\rho_{ABC} = \rho_{AB}\otimes \rho_C$. For such a state, we get
\begin{equation}
\begin{split}
    E_p^{(3)}(A:B:C) &= E^{(2)}_p(A:BC) = E^{(2)}_p(B:AC) = min\Big[S_{A\tilde{A}} + S_{B\tilde{B}}\Big],\\
    E_p^{(2)}(C:AB)&=0,\\
\end{split}
\end{equation}
and
\begin{equation}
\begin{split}
    E^{(3)}_p(A:B:\slashed{C}) =& min\Big[S_{A\tilde{A}} + S_{B\tilde{B}}\Big],\\
    E^{(3)}_p(A:\slashed{B}:C) =& min\Big[S_{A\tilde{A}}\Big],\\
    E^{(3)}_p(\slashed{A}:B:C) =& min\Big[S_{B\tilde{B}} \Big].\\
\end{split}
\end{equation}
Hence, 
\begin{equation}
    \hat{\Delta}^{(3)}_p(A:B:C) = min\Big[S_{A\tilde{A}} + S_{B\tilde{B}}\Big] - min\Big[S_{A\tilde{A}}\Big] -min\Big[S_{B\tilde{B}} \Big] = 0
\end{equation}
because $S_{A\tilde{A}} = S_{B\tilde{B}}$.

However, we do not have proof for $\hat{\Delta}^{(3)}_p(A:B:C)\geq 0$. We will study its properties in detail in the future.

\bibliographystyle{JHEP}
\bibliography{main}

@article{Umemoto-2018-MultipartiteEoP,
	author = {Umemoto, Koji and Zhou, Yang},
	date = {2018/10/24},
	doi = {10.1007/JHEP10(2018)152},
	id = {Umemoto2018},
	isbn = {1029-8479},
	journal = {Journal of High Energy Physics},
	number = {10},
	pages = {152},
	title = {Entanglement of purification for multipartite states and its holographic dual},
	url = {https://doi.org/10.1007/JHEP10(2018)152},
	volume = {2018},
	year = {2018}}

@article{Bao-2019-conditionalEoP,
  title = {Conditional and multipartite entanglements of purification and holography},
  author = {Bao, Ning and Halpern, Illan F.},
  journal = {Phys. Rev. D},
  volume = {99},
  issue = {4},
  pages = {046010},
  numpages = {7},
  year = {2019},
  month = {Feb},
  publisher = {American Physical Society},
  doi = {10.1103/PhysRevD.99.046010},
  url = {https://link.aps.org/doi/10.1103/PhysRevD.99.046010}
}

@article{Balasubramanian:2024ysu,
    author = "Balasubramanian, Vijay and Kang, Monica Jinwoo and Murdia, Chitraang and Ross, Simon F.",
    title = "{Signals of multiparty entanglement and holography}",
    eprint = "2411.03422",
    archivePrefix = "arXiv",
    primaryClass = "hep-th",
    month = "11",
    year = "2024"
}

@article{Iizuka:2025caq,
    author = "Iizuka, Norihiro and Lin, Simon and Nishida, Mitsuhiro",
    title = "{More on genuine multi-entropy and holography}",
    eprint = "2504.16589",
    archivePrefix = "arXiv",
    primaryClass = "hep-th",
    month = "4",
    year = "2025"
}

@article{Iizuka:2025ioc,
    author = "Iizuka, Norihiro and Nishida, Mitsuhiro",
    title = "{Genuine multi-entropy and holography}",
    eprint = "2502.07995",
    archivePrefix = "arXiv",
    primaryClass = "hep-th",
    month = "2",
    year = "2025"
}

@article{Gadde:2022cqi,
    author = "Gadde, Abhijit and Krishna, Vineeth and Sharma, Trakshu",
    title = "{New multipartite entanglement measure and its holographic dual}",
    eprint = "2206.09723",
    archivePrefix = "arXiv",
    primaryClass = "hep-th",
    reportNumber = "TIFR/TH/22-34",
    doi = "10.1103/PhysRevD.106.126001",
    journal = "Phys. Rev. D",
    volume = "106",
    number = "12",
    pages = "126001",
    year = "2022"
}

@article{Basak:2024uwc,
    author = "Basak, Jaydeep Kumar and Malvimat, Vinay and Yoon, Junggi",
    title = "{A New Genuine Multipartite Entanglement Measure: from Qubits to Multiboundary Wormholes}",
    eprint = "2411.11961",
    archivePrefix = "arXiv",
    primaryClass = "hep-th",
    month = "11",
    year = "2024"
}

@article{Gadde:2023zzj,
    author = "Gadde, Abhijit and Krishna, Vineeth and Sharma, Trakshu",
    title = "{Towards a classification of holographic multi-partite entanglement measures}",
    eprint = "2304.06082",
    archivePrefix = "arXiv",
    primaryClass = "hep-th",
    doi = "10.1007/JHEP08(2023)202",
    journal = "JHEP",
    volume = "08",
    pages = "202",
    year = "2023"
}

@article{Gadde:2023zni,
    author = "Gadde, Abhijit and Jain, Shraiyance and Krishna, Vineeth and Kulkarni, Harshal and Sharma, Trakshu",
    title = "{Monotonicity conjecture for multi-party entanglement. Part I}",
    eprint = "2308.16247",
    archivePrefix = "arXiv",
    primaryClass = "hep-th",
    doi = "10.1007/JHEP02(2024)025",
    journal = "JHEP",
    volume = "02",
    pages = "025",
    year = "2024"
}

@article{Gadde:2024jfi,
    author = "Gadde, Abhijit and Jain, Shraiyance and Kulkarni, Harshal",
    title = "{Multi-partite entanglement monotones}",
    eprint = "2406.17447",
    archivePrefix = "arXiv",
    primaryClass = "quant-ph",
    reportNumber = "TIFR/TH/24-10",
    month = "6",
    year = "2024"
}

@article{Gadde:2024taa,
    author = "Gadde, Abhijit and Harper, Jonathan and Krishna, Vineeth",
    title = "{Multi-invariants and Bulk Replica Symmetry}",
    eprint = "2411.00935",
    archivePrefix = "arXiv",
    primaryClass = "hep-th",
    reportNumber = "YITP-24-130, LCTP-24-19",
    month = "11",
    year = "2024"
}

@article{Harper:2024ker,
    author = "Harper, Jonathan and Takayanagi, Tadashi and Tsuda, Takashi",
    title = "{Multi-entropy at low Renyi index in 2d CFTs}",
    eprint = "2401.04236",
    archivePrefix = "arXiv",
    primaryClass = "hep-th",
    reportNumber = "YITP-24-02",
    doi = "10.21468/SciPostPhys.16.5.125",
    journal = "SciPost Phys.",
    volume = "16",
    number = "5",
    pages = "125",
    year = "2024"
}

@article{Ju:2024hba,
    author = "Ju, Xin-Xiang and Pan, Wen-Bin and Sun, Ya-Wen and Zhao, Yang",
    title = "{Holographic multipartite entanglement from the upper bound of $n$-partite information}",
    eprint = "2411.07790",
    archivePrefix = "arXiv",
    primaryClass = "hep-th",
    month = "11",
    year = "2024"
}

@article{Ju:2024kuc,
    author = "Ju, Xin-Xiang and Pan, Wen-Bin and Sun, Ya-Wen and Wang, Yuan-Tai and Zhao, Yang",
    title = "{More on the upper bound of holographic n-partite information}",
    eprint = "2411.19207",
    archivePrefix = "arXiv",
    primaryClass = "hep-th",
    doi = "10.1007/JHEP03(2025)184",
    journal = "JHEP",
    volume = "03",
    pages = "184",
    year = "2025"
}

@article{Ju:2025tgg,
    author = "Ju, Xin-Xiang and Sun, Ya-Wen and Zhao, Yang",
    title = "{Upper bound of holographic entanglement entropy combinations}",
    eprint = "2505.11059",
    archivePrefix = "arXiv",
    primaryClass = "hep-th",
    month = "5",
    year = "2025"
}

@article{Balasubramanian:2014hda,
    author = "Balasubramanian, Vijay and Hayden, Patrick and Maloney, Alexander and Marolf, Donald and Ross, Simon F.",
    title = "{Multiboundary Wormholes and Holographic Entanglement}",
    eprint = "1406.2663",
    archivePrefix = "arXiv",
    primaryClass = "hep-th",
    doi = "10.1088/0264-9381/31/18/185015",
    journal = "Class. Quant. Grav.",
    volume = "31",
    pages = "185015",
    year = "2014"
}

@article{Cui:2018dyq,
    author = "Cui, Shawn X. and Hayden, Patrick and He, Temple and Headrick, Matthew and Stoica, Bogdan and Walter, Michael",
    title = "{Bit Threads and Holographic Monogamy}",
    eprint = "1808.05234",
    archivePrefix = "arXiv",
    primaryClass = "hep-th",
    reportNumber = "BRX-6330, Brown-HET-1764, MIT-CTP-5036",
    doi = "10.1007/s00220-019-03510-8",
    journal = "Commun. Math. Phys.",
    volume = "376",
    number = "1",
    pages = "609--648",
    year = "2019"
}

@article{Harper:2019lff,
    author = "Harper, Jonathan and Headrick, Matthew",
    title = "{Bit threads and holographic entanglement of purification}",
    eprint = "1906.05970",
    archivePrefix = "arXiv",
    primaryClass = "hep-th",
    doi = "10.1007/JHEP08(2019)101",
    journal = "JHEP",
    volume = "08",
    pages = "101",
    year = "2019"
}

@article{Ma:2023ecg,
    author = "Ma, Mengru and Li, Yinfei and Shang, Jiangwei",
    title = "{Multipartite entanglement measures: a review}",
    eprint = "2309.09459",
    archivePrefix = "arXiv",
    primaryClass = "quant-ph",
    doi = "10.1016/j.fmre.2024.03.031",
    month = "9",
    year = "2023"
}

@article{quantum-entanglement-review,
  title = {Quantum entanglement},
  author = {Horodecki, Ryszard and Horodecki, Pawe\l{} and Horodecki, Micha\l{} and Horodecki, Karol},
  journal = {Rev. Mod. Phys.},
  volume = {81},
  issue = {2},
  pages = {865--942},
  numpages = {0},
  year = {2009},
  month = {Jun},
  publisher = {American Physical Society},
  doi = {10.1103/RevModPhys.81.865},
  url = {https://link.aps.org/doi/10.1103/RevModPhys.81.865}
}

@article{Swingle:2017blx,
    author = "Swingle, Brian",
    title = "{Spacetime from Entanglement}",
    doi = "10.1146/annurev-conmatphys-033117-054219",
    journal = "Ann. Rev. Condensed Matter Phys.",
    volume = "9",
    pages = "345--358",
    year = "2018"
}

@article{Maldacena:1997re,
    author = "Maldacena, Juan Martin",
    title = "{The Large $N$ limit of superconformal field theories and supergravity}",
    eprint = "hep-th/9711200",
    archivePrefix = "arXiv",
    reportNumber = "HUTP-97-A097, HUTP-98-A097",
    doi = "10.4310/ATMP.1998.v2.n2.a1",
    journal = "Adv. Theor. Math. Phys.",
    volume = "2",
    pages = "231--252",
    year = "1998"
}

@article{VanRaamsdonk:2010pw,
    author = "Van Raamsdonk, Mark",
    title = "{Building up spacetime with quantum entanglement}",
    eprint = "1005.3035",
    archivePrefix = "arXiv",
    primaryClass = "hep-th",
    doi = "10.1142/S0218271810018529",
    journal = "Gen. Rel. Grav.",
    volume = "42",
    pages = "2323--2329",
    year = "2010"
}

@article{Ryu:2006ef,
    author = "Ryu, Shinsei and Takayanagi, Tadashi",
    title = "{Aspects of Holographic Entanglement Entropy}",
    eprint = "hep-th/0605073",
    archivePrefix = "arXiv",
    reportNumber = "NSF-KITP-06-31, KUNS-2021",
    doi = "10.1088/1126-6708/2006/08/045",
    journal = "JHEP",
    volume = "08",
    pages = "045",
    year = "2006"
}

@article{Maldacena:2013xja,
    author = "Maldacena, Juan and Susskind, Leonard",
    title = "{Cool horizons for entangled black holes}",
    eprint = "1306.0533",
    archivePrefix = "arXiv",
    primaryClass = "hep-th",
    doi = "10.1002/prop.201300020",
    journal = "Fortsch. Phys.",
    volume = "61",
    pages = "781--811",
    year = "2013"
}

@article{Pastawski:2015qua,
    author = "Pastawski, Fernando and Yoshida, Beni and Harlow, Daniel and Preskill, John",
    title = "{Holographic quantum error-correcting codes: Toy models for the bulk/boundary correspondence}",
    eprint = "1503.06237",
    archivePrefix = "arXiv",
    primaryClass = "hep-th",
    doi = "10.1007/JHEP06(2015)149",
    journal = "JHEP",
    volume = "06",
    pages = "149",
    year = "2015"
}

@article{Brown:2015bva,
    author = "Brown, Adam R. and Roberts, Daniel A. and Susskind, Leonard and Swingle, Brian and Zhao, Ying",
    title = "{Holographic Complexity Equals Bulk Action?}",
    eprint = "1509.07876",
    archivePrefix = "arXiv",
    primaryClass = "hep-th",
    doi = "10.1103/PhysRevLett.116.191301",
    journal = "Phys. Rev. Lett.",
    volume = "116",
    number = "19",
    pages = "191301",
    year = "2016"
}

@article{Swingle:2014uza,
    author = "Swingle, Brian and Van Raamsdonk, Mark",
    title = "{Universality of Gravity from Entanglement}",
    eprint = "1405.2933",
    archivePrefix = "arXiv",
    primaryClass = "hep-th",
    month = "5",
    year = "2014"
}

@article{Swingle:2012wq,
    author = "Swingle, Brian",
    title = "{Constructing holographic spacetimes using entanglement renormalization}",
    eprint = "1209.3304",
    archivePrefix = "arXiv",
    primaryClass = "hep-th",
    month = "9",
    year = "2012"
}

@article{Bao:2015bfa,
    author = "Bao, Ning and Nezami, Sepehr and Ooguri, Hirosi and Stoica, Bogdan and Sully, James and Walter, Michael",
    title = "{The Holographic Entropy Cone}",
    eprint = "1505.07839",
    archivePrefix = "arXiv",
    primaryClass = "hep-th",
    reportNumber = "CALT-TH-2015-020, IPMU15-0074, SLAC-PUB-16294, SU-ITP-15-08",
    doi = "10.1007/JHEP09(2015)130",
    journal = "JHEP",
    volume = "09",
    pages = "130",
    year = "2015"
}

@article{Penington:2022dhr,
    author = "Penington, Geoff and Walter, Michael and Witteveen, Freek",
    title = "{Fun with replicas: tripartitions in tensor networks and gravity}",
    eprint = "2211.16045",
    archivePrefix = "arXiv",
    primaryClass = "hep-th",
    doi = "10.1007/JHEP05(2023)008",
    journal = "JHEP",
    volume = "05",
    pages = "008",
    year = "2023"
}

@article{Bao:2019zqc,
    author = "Bao, Ning and Cheng, Newton",
    title = "{Multipartite Reflected Entropy}",
    eprint = "1909.03154",
    archivePrefix = "arXiv",
    primaryClass = "hep-th",
    doi = "10.1007/JHEP10(2019)102",
    journal = "JHEP",
    volume = "10",
    pages = "102",
    year = "2019"
}

@article{Bao:2018fso,
    author = "Bao, Ning and Chatwin-Davies, Aidan and Remmen, Grant N.",
    title = "{Entanglement of Purification and Multiboundary Wormhole Geometries}",
    eprint = "1811.01983",
    archivePrefix = "arXiv",
    primaryClass = "hep-th",
    doi = "10.1007/JHEP02(2019)110",
    journal = "JHEP",
    volume = "02",
    pages = "110",
    year = "2019"
}

@article{Bao:2019bib,
    author = "Bao, Ning and Cao, ChunJun and Fischetti, Sebastian and Keeler, Cynthia",
    title = "{Towards Bulk Metric Reconstruction from Extremal Area Variations}",
    eprint = "1904.04834",
    archivePrefix = "arXiv",
    primaryClass = "hep-th",
    doi = "10.1088/1361-6382/ab377f",
    journal = "Class. Quant. Grav.",
    volume = "36",
    number = "18",
    pages = "185002",
    year = "2019"
}

@article{Cao:2020uvb,
    author = "Cao, ChunJun and Qi, Xiao-Liang and Swingle, Brian and Tang, Eugene",
    title = "{Building Bulk Geometry from the Tensor Radon Transform}",
    eprint = "2007.00004",
    archivePrefix = "arXiv",
    primaryClass = "hep-th",
    doi = "10.1007/JHEP12(2020)033",
    journal = "JHEP",
    volume = "12",
    pages = "033",
    year = "2020"
}

@article{Bao:2024azn,
    author = "Bao, Ning and Furuya, Keiichiro and Naskar, Joydeep",
    title = "{Towards a complete classification of holographic entropy inequalities}",
    eprint = "2409.17317",
    archivePrefix = "arXiv",
    primaryClass = "hep-th",
    doi = "10.1007/JHEP03(2025)117",
    journal = "JHEP",
    volume = "03",
    pages = "117",
    year = "2025"
}

@article{Hernandez-Cuenca:2022pst,
    author = "Hern\'andez-Cuenca, Sergio and Hubeny, Veronika E. and Rota, Massimiliano",
    title = "{The holographic entropy cone from marginal independence}",
    eprint = "2204.00075",
    archivePrefix = "arXiv",
    primaryClass = "hep-th",
    doi = "10.1007/JHEP09(2022)190",
    journal = "JHEP",
    volume = "09",
    pages = "190",
    year = "2022"
}

@article{Avis:2021xnz,
    author = "Avis, David and Hern\'andez-Cuenca, Sergio",
    title = "{On the foundations and extremal structure of the holographic entropy cone}",
    eprint = "2102.07535",
    archivePrefix = "arXiv",
    primaryClass = "math.CO",
    doi = "10.1016/j.dam.2022.11.016",
    journal = "Discrete Appl. Math.",
    volume = "328",
    pages = "16--39",
    year = "2023"
}

@article{Czech:2021rxe,
    author = "Czech, Bartlomiej and Shuai, Sirui",
    title = "{Holographic Cone of Average Entropies}",
    eprint = "2112.00763",
    archivePrefix = "arXiv",
    primaryClass = "hep-th",
    doi = "10.1038/s42005-022-01019-6",
    journal = "Commun. Phys.",
    volume = "5",
    pages = "244",
    year = "2022"
}

@article{Akers:2019gcv,
    author = "Akers, Chris and Rath, Pratik",
    title = "{Entanglement Wedge Cross Sections Require Tripartite Entanglement}",
    eprint = "1911.07852",
    archivePrefix = "arXiv",
    primaryClass = "hep-th",
    doi = "10.1007/JHEP04(2020)208",
    journal = "JHEP",
    volume = "04",
    pages = "208",
    year = "2020"
}

@article{Hayden:2021gno,
    author = "Hayden, Patrick and Parrikar, Onkar and Sorce, Jonathan",
    title = "{The Markov gap for geometric reflected entropy}",
    eprint = "2107.00009",
    archivePrefix = "arXiv",
    primaryClass = "hep-th",
    doi = "10.1007/JHEP10(2021)047",
    journal = "JHEP",
    volume = "10",
    pages = "047",
    year = "2021"
}

@article{Dutta:2019gen,
    author = "Dutta, Souvik and Faulkner, Thomas",
    title = "{A canonical purification for the entanglement wedge cross-section}",
    eprint = "1905.00577",
    archivePrefix = "arXiv",
    primaryClass = "hep-th",
    doi = "10.1007/JHEP03(2021)178",
    journal = "JHEP",
    volume = "03",
    pages = "178",
    year = "2021"
}

@article{Nguyen:2017yqw,
    author = "Nguyen, Phuc and Devakul, Trithep and Halbasch, Matthew G. and Zaletel, Michael P. and Swingle, Brian",
    title = "{Entanglement of purification: from spin chains to holography}",
    eprint = "1709.07424",
    archivePrefix = "arXiv",
    primaryClass = "hep-th",
    doi = "10.1007/JHEP01(2018)098",
    journal = "JHEP",
    volume = "01",
    pages = "098",
    year = "2018"
}

@article{AkersEoPinRTNs2024,
  title = {Entanglement of purification in random tensor networks},
  author = {Akers, Chris and Faulkner, Thomas and Lin, Simon and Rath, Pratik},
  journal = {Phys. Rev. D},
  volume = {109},
  issue = {10},
  pages = {L101902},
  numpages = {6},
  year = {2024},
  month = {May},
  publisher = {American Physical Society},
  doi = {10.1103/PhysRevD.109.L101902},
  url = {https://link.aps.org/doi/10.1103/PhysRevD.109.L101902}
}

@article{Akers2020EWCSTripartite,
	author = {Akers, Chris and Rath, Pratik},
	date = {2020/04/30},
	doi = {10.1007/JHEP04(2020)208},
	id = {Akers2020},
	isbn = {1029-8479},
	journal = {Journal of High Energy Physics},
	number = {4},
	pages = {208},
	title = {Entanglement wedge cross sections require tripartite entanglement},
	url = {https://doi.org/10.1007/JHEP04(2020)208},
	volume = {2020},
	year = {2020}}

@article{terhalEoP2002,
	author = {Terhal, Barbara M. and Horodecki, Micha{\l} and Leung, Debbie W. and DiVincenzo, David P.},
	doi = {10.1063/1.1498001},
	eprint = {https://pubs.aip.org/aip/jmp/article-pdf/43/9/4286/19183123/4286\_1\_online.pdf},
	issn = {0022-2488},
	journal = {Journal of Mathematical Physics},
	month = {09},
	number = {9},
	pages = {4286-4298},
	title = {The entanglement of purification},
	url = {https://doi.org/10.1063/1.1498001},
	volume = {43},
	year = {2002}}

@article{Umemoto2018EoPholography,
	author = {Umemoto, Koji and Takayanagi, Tadashi},
	date = {2018/06/01},
	doi = {10.1038/s41567-018-0075-2},
	id = {Umemoto2018},
	isbn = {1745-2481},
	journal = {Nature Physics},
	number = {6},
	pages = {573--577},
	title = {Entanglement of purification through holographic duality},
	url = {https://doi.org/10.1038/s41567-018-0075-2},
	volume = {14},
	year = {2018}}

@article{Harlow2017RTfomulaQEC,
	author = {Harlow, Daniel},
	date = {2017/09/01},
	doi = {10.1007/s00220-017-2904-z},
	id = {Harlow2017},
	isbn = {1432-0916},
	journal = {Communications in Mathematical Physics},
	number = {3},
	pages = {865--912},
	title = {The Ryu--Takayanagi Formula from Quantum Error Correction},
	url = {https://doi.org/10.1007/s00220-017-2904-z},
	volume = {354},
	year = {2017}}

@article{Pastawski2015Happy,
	author = {Pastawski, Fernando and Yoshida, Beni and Harlow, Daniel and Preskill, John},
	date = {2015/06/23},
	doi = {10.1007/JHEP06(2015)149},
	id = {Pastawski2015},
	isbn = {1029-8479},
	journal = {Journal of High Energy Physics},
	number = {6},
	pages = {149},
	title = {Holographic quantum error-correcting codes: toy models for the bulk/boundary correspondence},
	url = {https://doi.org/10.1007/JHEP06(2015)149},
	volume = {2015},
	year = {2015}}

@article{Almheiri2015bulklocality,
	author = {Almheiri, Ahmed and Dong, Xi and Harlow, Daniel},
	date = {2015/04/29},
	doi = {10.1007/JHEP04(2015)163},
	id = {Almheiri2015},
	isbn = {1029-8479},
	journal = {Journal of High Energy Physics},
	number = {4},
	pages = {163},
	title = {Bulk locality and quantum error correction in AdS/CFT},
	url = {https://doi.org/10.1007/JHEP04(2015)163},
	volume = {2015},
	year = {2015}}

@article{Faulkner:2020hzi,
    author = "Faulkner, Thomas",
    title = "{The holographic map as a conditional expectation}",
    eprint = "2008.04810",
    archivePrefix = "arXiv",
    primaryClass = "hep-th",
    month = "8",
    year = "2020"
}

@article{Furuya2022Petz,
	author = {Furuya, Keiichiro and Lashkari, Nima and Ouseph, Shoy},
	date = {2022/01/27},
	doi = {10.1007/JHEP01(2022)170},
	id = {Furuya2022},
	isbn = {1029-8479},
	journal = {Journal of High Energy Physics},
	number = {1},
	pages = {170},
	title = {Real-space RG, error correction and Petz map},
	url = {https://doi.org/10.1007/JHEP01(2022)170},
	volume = {2022},
	year = {2022}}

@article{Bao2019BeyondDistillation,
	author = {Bao, Ning and Penington, Geoffrey and Sorce, Jonathan and Wall, Aron C.},
	date = {2019/11/12},
	doi = {10.1007/JHEP11(2019)069},
	id = {Bao2019},
	isbn = {1029-8479},
	journal = {Journal of High Energy Physics},
	number = {11},
	pages = {69},
	title = {Beyond toy models: distilling tensor networks in full AdS/CFT},
	url = {https://doi.org/10.1007/JHEP11(2019)069},
	volume = {2019},
	year = {2019}}

@article{Mori2022distillation,
  title = {Entanglement distillation toward minimal bond cut surface in tensor networks},
  author = {Mori, Takato and Manabe, Hidetaka and Matsueda, Hiroaki},
  journal = {Phys. Rev. D},
  volume = {106},
  issue = {8},
  pages = {086008},
  numpages = {10},
  year = {2022},
  month = {Oct},
  publisher = {American Physical Society},
  doi = {10.1103/PhysRevD.106.086008},
  url = {https://link.aps.org/doi/10.1103/PhysRevD.106.086008}
}

@article{Mori:2024gwe,
    author = "Mori, Takato and Yoshida, Beni",
    title = "{Does connected wedge imply distillable entanglement?}",
    eprint = "2411.03426",
    archivePrefix = "arXiv",
    primaryClass = "hep-th",
    reportNumber = "YITP-24-149",
    month = "11",
    year = "2024"
}

@article{Gun2024distillation,
	author = {Bao, Ning and S{\"u}er, G{\"u}n},
	date = {2024/01/17},
	doi = {10.1007/JHEP01(2024)091},
	id = {Bao2024},
	isbn = {1029-8479},
	journal = {Journal of High Energy Physics},
	number = {1},
	pages = {91},
	title = {Holographic entanglement distillation from the surface state correspondence},
	url = {https://doi.org/10.1007/JHEP01(2024)091},
	volume = {2024},
	year = {2024}}

@article{Czech:2025jnw,
    author = "Czech, Bartlomiej and Shuai, Sirui and Wang, Yixu",
    title = "{Entropy Inequalities Constrain Holographic Erasure Correction}",
    eprint = "2502.12246",
    archivePrefix = "arXiv",
    primaryClass = "hep-th",
    month = "2",
    year = "2025"
}

@article{Caputa2019HEoPinCFT,
  title = {Holographic Entanglement of Purification from Conformal Field Theories},
  author = {Caputa, Pawel and Miyaji, Masamichi and Takayanagi, Tadashi and Umemoto, Koji},
  journal = {Phys. Rev. Lett.},
  volume = {122},
  issue = {11},
  pages = {111601},
  numpages = {6},
  year = {2019},
  month = {Mar},
  publisher = {American Physical Society},
  doi = {10.1103/PhysRevLett.122.111601},
  url = {https://link.aps.org/doi/10.1103/PhysRevLett.122.111601}
}

@article{Bhattacharyya2019freefields-spins-EoP,
  title = {Entanglement of Purification in Many Body Systems and Symmetry Breaking},
  author = {Bhattacharyya, Arpan and Jahn, Alexander and Takayanagi, Tadashi and Umemoto, Koji},
  journal = {Phys. Rev. Lett.},
  volume = {122},
  issue = {20},
  pages = {201601},
  numpages = {6},
  year = {2019},
  month = {May},
  publisher = {American Physical Society},
  doi = {10.1103/PhysRevLett.122.201601},
  url = {https://link.aps.org/doi/10.1103/PhysRevLett.122.201601}
}

@article{Bhattacharyya2018freescalarEoP,
	author = {Bhattacharyya, Arpan and Takayanagi, Tadashi and Umemoto, Koji},
	date = {2018/04/24},
	doi = {10.1007/JHEP04(2018)132},
	id = {Bhattacharyya2018},
	isbn = {1029-8479},
	journal = {Journal of High Energy Physics},
	number = {4},
	pages = {132},
	title = {Entanglement of purification in free scalar field theories},
	url = {https://doi.org/10.1007/JHEP04(2018)132},
	volume = {2018},
	year = {2018}}

@article{GUO2019134934,
	author = {Wu-Zhong Guo},
	doi = {https://doi.org/10.1016/j.physletb.2019.134934},
	issn = {0370-2693},
	journal = {Physics Letters B},
	pages = {134934},
	title = {Entanglement of purification and projection operator in conformal field theories},
	url = {https://www.sciencedirect.com/science/article/pii/S0370269319306562},
	volume = {797},
	year = {2019}}

@article{Akers2023RE-RTNsII,
	author = {Akers, Chris and Faulkner, Thomas and Lin, Simon and Rath, Pratik},
	date = {2023/01/13},
	doi = {10.1007/JHEP01(2023)067},
	id = {Akers2023},
	isbn = {1029-8479},
	journal = {Journal of High Energy Physics},
	number = {1},
	pages = {67},
	title = {Reflected entropy in random tensor networks. Part II. A topological index from canonical purification},
	url = {https://doi.org/10.1007/JHEP01(2023)067},
	volume = {2023},
	year = {2023}}

@article{Akers2022RE-RTNs,
	author = {Akers, Chris and Faulkner, Thomas and Lin, Simon and Rath, Pratik},
	date = {2022/05/24},
	doi = {10.1007/JHEP05(2022)162},
	id = {Akers2022},
	isbn = {1029-8479},
	journal = {Journal of High Energy Physics},
	number = {5},
	pages = {162},
	title = {Reflected entropy in random tensor networks},
	url = {https://doi.org/10.1007/JHEP05(2022)162},
	volume = {2022},
	year = {2022}}

@article{Akers2024RE-RTNsIII,
	author = {Akers, Chris and Faulkner, Thomas and Lin, Simon and Rath, Pratik},
	date = {2024/12/27},
	doi = {10.1007/JHEP12(2024)209},
	id = {Akers2024},
	isbn = {1029-8479},
	journal = {Journal of High Energy Physics},
	number = {12},
	pages = {209},
	title = {Reflected entropy in random tensor networks. Part III. Triway cuts},
	url = {https://doi.org/10.1007/JHEP12(2024)209},
	volume = {2024},
	year = {2024}}

@article{Hayden2013MMI,
  title = {Holographic mutual information is monogamous},
  author = {Hayden, Patrick and Headrick, Matthew and Maloney, Alexander},
  journal = {Phys. Rev. D},
  volume = {87},
  issue = {4},
  pages = {046003},
  numpages = {11},
  year = {2013},
  month = {Feb},
  publisher = {American Physical Society},
  doi = {10.1103/PhysRevD.87.046003},
  url = {https://link.aps.org/doi/10.1103/PhysRevD.87.046003}
}

@article{Cui2020BitThreads,
	author = {Cui, Shawn X. and Hayden, Patrick and He, Temple and Headrick, Matthew and Stoica, Bogdan and Walter, Michael},
	date = {2020/05/01},
	doi = {10.1007/s00220-019-03510-8},
	id = {Cui2020},
	isbn = {1432-0916},
	journal = {Communications in Mathematical Physics},
	number = {1},
	pages = {609--648},
	title = {Bit Threads and Holographic Monogamy},
	url = {https://doi.org/10.1007/s00220-019-03510-8},
	volume = {376},
	year = {2020}}

@article{Chandra2023RTN-HoloCodes,
	author = {Chandra, Jeevan and Hartman, Thomas},
	date = {2023/05/15},
	doi = {10.1007/JHEP05(2023)109},
	id = {Chandra2023},
	isbn = {1029-8479},
	journal = {Journal of High Energy Physics},
	number = {5},
	pages = {109},
	title = {Toward random tensor networks and holographic codes in CFT},
	url = {https://doi.org/10.1007/JHEP05(2023)109},
	volume = {2023},
	year = {2023}}

@article{Verlinde2013BH-QEC,
	author = {Verlinde, Erik and Verlinde, Herman},
	date = {2013/10/17},
	doi = {10.1007/JHEP10(2013)107},
	id = {Verlinde2013},
	isbn = {1029-8479},
	journal = {Journal of High Energy Physics},
	number = {10},
	pages = {107},
	title = {Black hole entanglement and quantum error correction},
	url = {https://doi.org/10.1007/JHEP10(2013)107},
	volume = {2013},
	year = {2013}}

@article{Kar2023Non-isometric,
	author = {Kar, Arjun},
	date = {2023/02/20},
	doi = {10.1007/JHEP02(2023)195},
	id = {Kar2023},
	isbn = {1029-8479},
	journal = {Journal of High Energy Physics},
	number = {2},
	pages = {195},
	title = {Non-isometric quantum error correction in gravity},
	url = {https://doi.org/10.1007/JHEP02(2023)195},
	volume = {2023},
	year = {2023}}

@article{Leutheusser2025subsubduality,
  title = {Subregion-subalgebra duality: Emergence of space and time in holography},
  author = {Leutheusser, Sam and Liu, Hong},
  journal = {Phys. Rev. D},
  volume = {111},
  issue = {6},
  pages = {066021},
  numpages = {51},
  year = {2025},
  month = {Mar},
  publisher = {American Physical Society},
  doi = {10.1103/PhysRevD.111.066021},
  url = {https://link.aps.org/doi/10.1103/PhysRevD.111.066021}
}

@article{Faulkner:2022ada,
    author = "Faulkner, Thomas and Li, Min",
    title = "{Asymptotically isometric codes for holography}",
    eprint = "2211.12439",
    archivePrefix = "arXiv",
    primaryClass = "hep-th",
    month = "11",
    year = "2022"
}

@article{Berthiere:2025toi,
    author = "Berthi{\`e}re, Cl{\'e}ment and Gaudin, Paul",
    title = "{Genuine multi-entropy, dihedral invariants and Lifshitz theory}",
    eprint = "2509.00593",
    archivePrefix = "arXiv",
    primaryClass = "hep-th",
    month = "8",
    year = "2025"
}

@article{Harper:2025uui,
    author = "Harper, Jonathan and Mollabashi, Ali and Takayanagi, Tadashi and Tasuki, Kenya",
    title = "{Multi-entropy and the Dihedral Measures at Quantum Critical Points}",
    eprint = "2506.10396",
    archivePrefix = "arXiv",
    primaryClass = "hep-th",
    reportNumber = "YITP-25-87",
    month = "6",
    year = "2025"
}

@article{Balasubramanian:2025hxg,
    author = "Balasubramanian, Vijay and Kang, Monica Jinwoo and Cummings, Charlie and Murdia, Chitraang and Ross, Simon F.",
    title = "{Purely GHZ-like entanglement is forbidden in holography}",
    eprint = "2509.03621",
    archivePrefix = "arXiv",
    primaryClass = "hep-th",
    month = "9",
    year = "2025"
}

@article{BFM:preparation,
    author = "Bao, Ning and Furuya, Keiichiro and March, Jacob",
    title = "",
    eprint = "{In progress}"
}

@article{Freivogel2016gaugeQECadscft,
	author = {Freivogel, Ben and Jefferson, Robert A. and Kabir, Laurens},
	date = {2016/04/19},
	doi = {10.1007/JHEP04(2016)119},
	id = {Freivogel2016},
	isbn = {1029-8479},
	journal = {Journal of High Energy Physics},
	number = {4},
	pages = {119},
	title = {Precursors, gauge invariance, and quantum error correction in AdS/CFT},
	url = {https://doi.org/10.1007/JHEP04(2016)119},
	volume = {2016},
	year = {2016}}

@article{Hayden2023REnotCorrelation,
  title = {Reflected entropy: Not a correlation measure},
  author = {Hayden, Patrick and Lemm, Marius and Sorce, Jonathan},
  journal = {Phys. Rev. A},
  volume = {107},
  issue = {5},
  pages = {L050401},
  numpages = {4},
  year = {2023},
  month = {May},
  publisher = {American Physical Society},
  doi = {10.1103/PhysRevA.107.L050401},
  url = {https://link.aps.org/doi/10.1103/PhysRevA.107.L050401}
}

@article{Louisia:2025bxz,
    author = "Louisia, Kyan and Mori, Takato and Warner, Herbie",
    title = "{Zoo of Correlation Inequalities in Holography and Beyond}",
    eprint = "2511.21870",
    archivePrefix = "arXiv",
    primaryClass = "hep-th",
    reportNumber = "RUP-25-24",
    month = "11",
    year = "2025"
}

@article{Mori:2025gqe,
    author = "Mori, Takato",
    title = "{Quantum correlation beyond entanglement: Holographic discord and multipartite generalizations}",
    eprint = "2506.02131",
    archivePrefix = "arXiv",
    primaryClass = "hep-th",
    reportNumber = "RUP-25-11, YITP-25-66",
    month = "6",
    year = "2025"
}

\end{document}